\documentclass[final]{lmcs}
\pdfoutput=1

\usepackage{lastpage}
\renewcommand{\dOi}{14(3:22)2018}
\lmcsheading{}{1--\pageref{LastPage}}{}{}%
{Oct.~27,~2017}{Sep.~24,~2018}{}

\usepackage[utf8]{inputenc}
\usepackage[T1]{fontenc}

\newcommand{\takeout}[1]{\empty}

\usepackage{ifthen}
\newboolean{full}
\setboolean{full}{true}

\usepackage[notcite,notref]{showkeys}

\usepackage{seqsplit}
\usepackage{xstring}

\makeatletter
\renewcommand*\showkeyslabelformat[1]{%
\@ifundefined{hideNextShowKeysLabel}{%
\noexpandarg%
\StrSubstitute{#1}{ }{\textvisiblespace}[\TEMP]%
\parbox[t]{\marginparwidth}{\raggedright\normalfont\small\ttfamily\{{\color{red!50!black}\expandafter\seqsplit\expandafter{\TEMP}}\}}%
}{}% end of \@ifundefined
}
\makeatother

\usepackage[layout=footnote]{fixme}

\usepackage{etex}

\usepackage{macros}

\FXRegisterAuthor{sm}{asm}{SM}%Stefan

\begin{document}

%
% Frontmatter
%
\title[Proper Functors]{Proper Functors and Fixed Points for Finite Behaviour}

\author[S.~Milius]{Stefan Milius}
\address{Lehrstuhl f\"{u}r Informatik 8 (Theoretische Informatik),
  \newline Friedrich-Alexander-Universit\"{a}t Erlangen-N\"{u}rnberg, Germany}
\email{mail@stefan-milius.eu}
\thanks{Supported by the Deutsche Forschungsgemeinschaft (DFG) under project MI~717/5-1}

%\subjclass{F.3.2~Semantics of Programming Languages}


\begin{abstract}
  \smnote[inline]{Maybe reformulate to include universal property of
    $\phi F$.}
  The rational fixed point of a set functor is well-known to capture
  the behaviour of finite coalgebras. In this paper we consider
  functors on algebraic categories. For them the rational fixed point
  may no longer be fully abstract, i.e.~a subcoalgebra of the final
  coalgebra. Inspired by \'Esik and Maletti's notion of a proper
  semiring, we introduce the notion of a proper functor. We show that
  for proper functors the rational fixed point is determined as the
  colimit of all coalgebras with a free finitely generated algebra as
  carrier and it is a subcoalgebra of the final coalgebra. Moreover,
  we prove that a functor is proper if and only if that colimit is a
  subcoalgebra of the final coalgebra. These results serve as
  technical tools for soundness and completeness proofs for
  coalgebraic regular expression calculi, e.g.~for weighted automata.
\end{abstract}

\maketitle
\smnote{This is a new title!}

%
% end of frontmatter - start of paper
%
\section{Introduction}

\takeout{% Stoffsammlung
\begin{itemize}
\item rational fixed point as fully abstract 
\item locally finite fixed point as subcoalgebra 
\item picture with 4 fixed points (coalgebras) and hierarchy 
\end{itemize}}% end takeout

Coalgebras allow to model many types of systems within a uniform and
conceptually clear mathematical framework~\cite{rutten}. One of the
key features of this framework is \emph{final semantics}; the final
coalgebra provides a fully abstract domain of system behaviour
(i.e.~it identifies precisely the behaviourally equivalent
states). For example, the standard coalgebraic modelling of
deterministic automata (without restricting to finite state sets)
yields the set of formal languages as final coalgebra. Restricting to
finite automata, one obtains precisely the regular
languages~\cite{Rutten:1998:ACE}. It is well-known that this
correspondence can be generalized to locally finitely presentable
(lfp) categories~\cite{ar}, where \emph{finitely presentable} objects
play the role of finite sets. For a finitary functor $F$ (modelling a
coalgebraic system type) one then obtains the \emph{rational fixed
  point} $\rho F$, which provides final semantics to all coalgebras
with a finitely presentable carrier~\cite{m_linexp}. Moreover, the
rational fixed point is \emph{fully abstract}, i.e.~$\rho F$ is a
subcoalgebra of the final one $\nu F$, whenever the classes of
finitely presentable and finitely generated objects agree in the base
category and $F$ preserves non-empty
monomorphisms~\cite[Section~5]{mpw17}. While the latter
assumption on $F$ is very mild, the former one on the base category is more
restrictive. However, it is still true for many categories used in the
construction of coalgebraic system models (e.g.~sets, posets, graphs,
vector spaces, commutative monoids, nominal sets and positively convex
algebras). Hence, in these cases the rational fixed point $\rho F$ is
the canonical domain of \emph{regular} behaviour, i.e.~the behaviour
of `finite' systems of type $F$. 

In this paper we will consider rational fixed points in algebraic
categories (a.k.a.~finitary varieties), i.e.~categories of algebras
specified by a signature of operation symbols with finite arity and a set of
equations (equivalently, these are precisely the Eilenberg-Moore
categories for finitary monads on sets). Being the target of
generalized determinization~\cite{sbbr13}, these categories provide a
paradigmatic setting for coalgebraic modelling beyond sets. For
example, non-deterministic automata, weighted or probabilistic
automata~\cite{JacobsEA15}, or context-free grammars~\cite{wbr13} are
coalgebraically modelled over the categories of join-semilattices,
modules for a semiring, positively convex algebras, and idempotent semirings,
respectively. In algebraic categories one would like that the rational
fixed point, in addition to being fully abstract, is determined
already by those coalgebras carried by free finitely generated
algebras, i.e.~precisely those coalgebras arising by generalized
determinization. In particular, this feature is used in completeness
proofs for generalized regular expressions
calculi~\cite{brs_lmcs,sbbr13,bms13}; there one proves that the
quotient of syntactic expressions modulo axioms of the calculus is
(isomorphic to) the rational fixed point by establishing its universal
property as a final object for that quotient. A key feature of the settings in
loc.~cit.~is that it suffices to verify the finality only
w.r.t.~coalgebras with a free finitely generated carrier.

The purpose of the present paper is to provide sufficient conditions
on the algebraic base category and coalgebraic type functor that
ensure that the rational fixed point is fully abstract and that such
finality proofs are sound. To this end we form a coalgebra that serves
as the semantics domain of all behaviours of target coalgebras of
generalized determinization (modulo bisimilarity on the level of these
coalgebra). More precisely, let $T: \Set \to \Set$ be a finitary monad
on sets and $F: \Set^T \to \Set^T$ be a finitary endofunctor
preserving surjective $T$-algebra morphisms (note that the last
assumption always holds if $F$ is lifted from some endofunctor on
$\Set$). Now form the colimit $\phi F$ of the inclusion functor of the
full subcategory $\coafr F$ formed by all $F$-coalgebras of the form
$TX \to FTX$, where $X$ is a finite set. Urbat~\cite{Urbat17} has
shown that $\phi F$ is a fixed point of $F$. We first first provide a
characterization of $\phi F$ that uniquely determines it up to
isomorphism: based on Ad\'amek et al.'s notion of a Bloom
algebra~\cite{ahm14}, we introduce the new notion of an
\emph{ffg-Bloom algebra}, and we prove that, considered as an algebra
for $F$, $\phi F$ is the initial ffg-Bloom algebra
(Theorem~\ref{thm:Bloomini}).

Then we turn to the full abstractness of the rational fixed point
$\rho F$ and the soundness of the above mentioned finality
proofs. Inspired by \'Esik and Maletti's notion of a proper semiring
(which is in fact a notion concerning weighted automata), we introduce
\emph{proper functors} (Definition~\ref{dfn:proper}), and we prove
that for a proper functor on an algebraic category the rational fixed
point is determined by the coalgebras with a free finitely generated
carrier. More precisely, if $F$ is proper, then the rational fixed
point $\rho F$ is (isomorphic to) initial Bloom algebra $\phi
F$. Moreover, we show that a functor $F$ is proper if and only if
$\phi F$ is a subcoalgebra of the final coalgebra $\nu F$
(Theorem~\ref{thm:main}). As a consequence we also obtain the desired
result that for a proper functor $F$ the finality property of $\rho F$ can
be established by only verifying that property for all coalgebras from
$\coafr F$ (Corollary~\ref{cor:final}).

In addition, we provide more easily established sufficient conditions on
$\Set^T$ and $F$ that ensure properness: $F$ is proper if finitely
generated algebras of $\Set^T$ are closed under kernel pairs and $F$
maps kernel pairs to weak pullbacks in $\Set$. For a lifting $F$ this
holds whenever the lifted functor on sets preserves weak pullbacks; in
fact, in this case the above conditions were shown to entail
Corollary~\ref{cor:final} in previous
work~\cite[Corollary~3.36]{bms13}. However, the type functor (on the
category of commutative monoids) of weighted automata with weights
drawn from the semiring of natural numbers provides an example of a
proper functor for which the above condition on $\Set^T$ fails.

Another recent related work concerns the so-called \emph{locally
  finite fixed point} $\theta F$~\cite{mpw17}; this provides a fully
abstract behavioural domain whenever $F$ is a finitary endofunctor on
an lfp category preserving non-empty monomorphisms. In loc.~cit.~it was shown
that $\theta F$ captures a number of instances that cannot be captured
by the rational fixed point, e.g.~context free languages~\cite{wbr13},
constructively algebraic formal power-series~\cite{PetreS09,wbr15},
Courcelle's algebraic trees~\cite{courcelle,amv_secalg} and the
behaviour of stack machines~\cite{gms14}. However, as far as we know,
$\theta F$ is not amenable to the simplified finality check mentioned
above unless $F$ is proper.
%In general, in
%an algebraic category these form a larger class of objects than
%finitely presentable ones, which in turn form a larger class than free
%finitely generated ones.

Putting everything together, in an algebraic category we obtain the
following picture of fixed points of $F$ (where $\epito$ denotes
quotient coalgebras and $\monoto$ a subcoalgebra):
\begin{equation}\label{diag:hierarchy}
  \phi F \epito \rho F \epito \theta F \monoto \nu F.
\end{equation}
We exhibit an example, where all four fixed points are
different. However, if $F$ is proper and preserves monomorphisms, then
$\phi F$, $\rho F$ and $\theta F$ are isomorphic and fully abstract,
i.e.~they collapse to a subcoalgebra of the final one:
$\phi F \cong \rho F \cong \theta F \monoto \nu F$.

At this point, note that Urbat's above mentioned recent
work~\cite{Urbat17} also provides a framework which covers the four
above fixed points as four instances of one theory. This provides, for
example, a uniform proof of the fact that they are fixed points and
their universal properties (in the case of $\rho F$, $\theta F$ and
$\phi F$). However, Urbat's paper does not study the relationship
between the four fixed points.

The rest of the paper is structured as follows: in
Section~\ref{sec:prelim} we collect some technical preliminaries and
recall the rational and locally finite fixed points more in
detail. Section~\ref{sec:phi} introduces the new fixed point $\phi F$
and establishes the picture in~\eqref{diag:hierarchy}. Next,
Section~\ref{sec:univ} provides the characterization of $\phi F$ as
the initial ffg-Bloom algebra for $F$. Section~\ref{sec:prop}
introduces proper functors and presents our main results, while in
Section~\ref{sec:proof} we present the proof of
Theorem~\ref{thm:main}. Finally, Section~\ref{sec:con} concludes the
paper.

This paper is a reworked full version of the conference
paper~\cite{milius17}. We have included detailed proofs, and in
addition, we have added the new results in Section~\ref{sec:univ}.

\paragraph{\bf Acknowledgments} I would like to thank Ji\v{r}\'i
Ad\'amek, Henning Urbat and Joost Winter for helpful discussions. I am
also grateful to the anonymous reviewers whose constructive comments
have helped to improve the presentation of this paper.

\section{Preliminaries}
\label{sec:prelim}

In this section we recall a few preliminaries needed for the
subsequent development. We assume that readers are familiar with basic
concepts of category theory. 

We denote the coproduct of two objects $X$ and $Y$ of a category $\A$
by $X+Y$ with injections $\inl: X \to X+Y$ and $\inr: Y \to X+Y$. 
%, and we write $\can$ for the canonical morphism
%\[
%  \can = [F\inl, F\inr]: FX + FY \to F(X+Y)
%\]
%for every endofunctor $F$ on $\A$. 

\begin{rem}
  Recall that a \emph{strong epimorphism} in a category $\A$ is an
  epimorphism $e: A \epito B$ of $\A$ that has the \emph{unique
  diagonal property} w.r.t.~any monomorphism. More precisely,
  whenever the outside of the following square 
  \[
    \xymatrix{
      A \ar@{->>}[r]^-e \ar[d]_f 
      & B \ar[d]^g \ar@{-->}[ld]_d \\
      C \ar@{ >->}[r]_-m  & D}
  \]
  commutes, where $m:C \monoto D$ is a
  monomorphism, then there exists a unique morphism $d: B \to C$ with
  $d \cdot e = f$ and $m \cdot d = g$. \iffull

  Similarly, a jointly epimorphic family $e_i: A_i \to B$, $i \in I$,
  is \emph{strong} if it has the following similar unique diagonal property: for
  every monomorphism $m: C \monoto D$ and morphisms $g: B \to D$ and
  $f_i: A_i \to C$, $i \in I$, such that $m \cdot f_i = g \cdot e_i$
  holds for all $i \in I$, there exists a unique $d: C \to D$ such
  that $m \cdot d = g$ and $d \cdot e_i = f_i$ for all $i \in I$.\fi
\end{rem}

\iffull On several occasions we will make use of the following fact.
\begin{lem}\label{lem:strepi}
  Let $D: \D \to \C$ be a diagram with a colimit cocone $\inj_d: Dd
  \to C$. Then the colimit injections $\inj_d$ form a strongly
  epimorphic family. 
\end{lem}  
\begin{proof}
  First, it is easy to see that the $\inj_d$ form a jointly epimorphic
  family. To see that it is strong, suppose we have a monomorphism
  $m: M \monoto N$ and morphisms $g: C \to N$ and $f_d: Dd \to M$ for
  every object $d$ in $\D$ such that $m \cdot f_d = g \cdot
  \inj_d$. Then the $f_d: Dd \to M$ form a cocone of $D$. Indeed,
  for every morphism $h: d \to d'$ of $\D$ we have
  \[
    m \cdot f_{d'} \cdot Dh = g \cdot \inj_{d'} \cdot Dh = g \cdot \inj_{d}
    = m \cdot f_{d}, 
  \]
  which implies that $f_{d'} \cdot Dh = f_d$ since $m$ is a
  monomorphism. Therefore there exists a unique $i: C \to M$ such that
  $f_d = i \cdot \inj_d$ for every $d$ in $\D$. It follows that also
  $m \cdot i = g$ since this equation holds when extended by every
  $\inj_d$; then use that the $\inj_d$ form an epimorphic family.
\end{proof}
\fi% end full version

\subsection{Algebras and Coalgebras}
\label{sec:algcoalg}
\takeout{% Stoffsammlung taken out
\begin{itemize}
\item monad with unit $\eta$ and multiplication $\mu$; Kleisli-triple
  $(T,\eta, \kl{(-)})$
\item algebras for a monad, Eilenberg-Moore category $\Set^T$
\item just write the carrier
\item coalgebras for a functor $F$; homomorphisms; final coalgebra
  $(\nu F,t)$ and behavioural equivalence $\sim$; notation $\fin c: C
  \to \nu F$
  for the unique coalgebra morphism.

\item leading example are weighted automata; so recall semirings,
  weighted automata, weighted languages, Noetherian semirings and examples 
\end{itemize}}% end takeout

We also assume that readers are familiar with algebras and coalgebras for
an endofunctor. Given an endofunctor $F$ on some category $\A$ we
write $(\nu F,t)$ for the final $F$-coalgebra (if it exists). Recall, that
the final $F$-coalgebra exists under mild assumptions on $\A$ and $F$,
e.g.~whenever $\A$ is locally presentable and $F$ an accessible
functor (see~\cite{ar}). For any coalgebra $c: C \to FC$ we will write
$\fin c: C \to \nu F$ for the unique coalgebra morphism. We write
\[
  \coa F
\]
for the category of $F$-coalgebras and their morphisms. Recall that
all colimits in $\coa F$ are formed on the level of $\A$, i.e.~the
canonical forgetful functor $\coa F \to \A$ creates all colimits
(see e.g.~\cite[Prop.~4.3]{adamek_survey}).

If $\A$ is a concrete category, i.e.~equipped with a faithful functor
$\undl \cdot: \A \to \Set$, one defines \emph{behavioural equivalence}
as the following relation $\sim$: given two $F$-coalgebras $(X,c)$ and
$(Y,d)$ then $x \sim y$ holds for $x \in \undl X$ and $y\in \undl Y$ if
there is another $F$-coalgebra $(Z,e)$ and $F$-coalgebra morphisms $f:
X \to Z$ and $g: Y \to Z$ with $\undl f(x) = \undl g(y)$. 

The base categories $\A$ of interest in this paper are the
\emph{algebraic categories}, i.e.~categories of Eilenberg-Moore
algebras (or $T$-algebras, for short) for a finitary monad $T$ on
$\Set$. Recall that, for a monad $T$ on $\Set$ with unit $\eta$ and
multiplication $\mu$, a $T$-algebra is a pair $(A,\alpha)$ where
$\alpha: TA \to A$, called the \emph{algebra structure}, is a map such
that the diagram below commutes:
\[
  \xymatrix{
    A \ar[r]^-{\eta_A} \ar@{=}[rd] & TA \ar[d]^\alpha & TTA \ar[l]_-{\mu_A}
    \ar[d]^{T\alpha} \\
    & A & TA \ar[l]^-\alpha
  }
\]
Morphisms of $T$-algebras are just the usual morphisms of functor
algebra, i.e.~a $T$-algebra morphism $h: (A,\alpha) \to (B,\beta)$ is a
map $h: A \to B$ such that the square below commutes:
\[
  \xymatrix{
    TA \ar[r]^-\alpha \ar[d]_{Th} & A \ar[d]^h \\
    TB \ar[r]_-\beta & B
    }
\]
The category of $T$-algebras and their morphisms is denoted by
$\Set^T$ as usual. Equivalently, those categories are precisely the finitary
varieties, i.e.~category of $\Sigma$-algebras for a signature
$\Sigma$, whose operation symbols have finite arity,
satisfying a set of equations (e.g.~the categories of
monoids, groups, vector spaces, or join-semilattices).

We will frequently make use of the fact that $(TX, \mu_X)$ is the
\emph{free} $T$-algebra on the set $X$ (of generators). This means
that for every $T$-algebra $(A, \alpha)$ and every map $f: X \to A$
there exists a unique extension of $f$ to a $T$-algebra morphisms,
i.e., there exists a unique $T$-algebra morphism $\kl f$ such that
$\kl f \cdot \eta_X = f$:
\[
  \xymatrix{
    X \ar[r]^{\eta_X} \ar[rd]_f & TX \ar[d]^{\kl f} & TTX \ar[l]_-{\mu_X}
    \ar[d]^{T\kl f}
    \\
    & A & TA \ar[l]^-\alpha
    }
\]
Moreover, it is easy to verify that $\kl f = \mu_A \cdot Tf$ holds.

A free $T$-algebra $(TX, \mu_X)$ where $X$ is a
finite set is called \emph{free finitely generated}.  

In the following we will often drop algebra structures when we
discuss a $T$-algebra $(A,\alpha)$ and simply speak of the algebra $A$. 

\begin{exa}
  \label{ex:leading}
  \begin{enumerate}
  \item  The leading example in this paper are weighted automata considered
  as coalgebras. Let $(\S, +,\cdot , 0, 1)$ be a semiring,
  i.e.~$(\S, +, 0)$ is a commutative monoid, $(\S, \cdot, 1)$ a monoid
  and the usual distributive laws hold: $r \o 0 = 0 = 0 \o r$,
  $r\o (s + t) = r \o s + r \o t$ and $(r + s)\o t = r \o t + s \o
  t$. We just write $\S$ to denote a semiring. As base category $\A$
  we consider the category $\smod$ of $\S$-semimodules; recall that a
  (left) \emph{$\S$-semimodule} is a commutative monoid $(M, +, 0)$
  together with an action $\S \times M \to M$, written as
  juxtaposition $sm$ for $r \in \S$ and $m \in M$, such that for every
  $r,s \in \S$ and every $m, n \in M$ the following laws hold:
  $$
  \begin{array}{r@{\ }c@{\ }l@{\qquad}r@{\ }c@{\ }l@{\qquad}r@{\ }c@{\ }l}
    (r+s)m & = & rm + sm & 0m &  = & 0 & 1m &  = & m \\
    r(m+n) & = & rm + rn & r0 & = & 0 & r(sm) & = & (r \o s) m
  \end{array}
  $$
  An $\S$-semimodule morphism is a monoid homomorphism
  $h\colon M_1 \to M_2$ such that $h(rm) = rh(m)$ for each $r \in \S$
  and $m \in M_1$.

  An $\S$-weighted automaton over the fixed input alphabet $\Sigma$ is a
  triple $(i, (M^a)_{a \in \Sigma}, o)$, where $i$ and $o$ are a row and a
  column vector in $\S^n$, respectively, of input and output weights,
  respectively, and $M_a$ is an $n\times n$-matrix over $\S$, for some
  natural number $n$. This number $n$ is the number of states of the
  weighted automaton and the matrices $M^a$ represent $\S$-weighted
  transitions; in fact, $M^a_{i,j}$ is the weight of the
  $a$-transition from state $i$ to state $j$ (with a weight of $0$
  meaning that there is no $a$-transition). Every weighted automaton
  accepts a \emph{formal power series} (or \emph{weighted language})
  $L: \Sigma^* \to \S$ defined in the following way:
  $L(w) = i \cdot M^w \cdot o$ where $M^w$ is the obvious inductive
  extension of $a \mapsto M^a$ to words in $\Sigma^*$: $M^\eps$ is the
  identity matrix and $M^{av} = M^a \cdot M^v$ for every $a \in \Sigma$ and
  $v \in \Sigma^*$.
  
  Now consider the functor $FX = \S \times X^\Sigma$ on
  $\smod$. Clearly, a weighted automaton (without its initial vector)
  on $n$ states is equivalently an $F$-coalgebra on $\S^n$; in fact,
  to give a coalgebra structure $\S^n \to \S \times (\S^n)^\Sigma$ amount
  to specifying two $\S$-semimodule morphisms $o: \S^n \to \S$
  (equivalently, a column vector over in $\S$) and
  $t: \S^n \to (\S^n)^\Sigma$ (equivalently, an $\Sigma$-indexed family of
  $\S$-semimodule morphisms on $\S^n$ each of which can be represented
  by an $n\times n$-matrix).

  The final $F$-coalgebra is carried by the set $\S^{\Sigma^*}$ of all
  weighted languages over $\Sigma$
  with the obvious (coordinatewise) $\S$-semimodule structure and with
  the $F$-coalgebra structure given by
  $\langle o, t\rangle: \S^{\Sigma^*} \to \S \times (\S^{\Sigma^*})^\Sigma$ with
  $o(L) = L(\eps)$ and $t(L)(a) = \lambda w. L(aw)$; it is
  straightforward to verify that $o$ and $t$ are $\S$-semimodule
  morphisms and form a final coalgebra. Moreover, for every
  $F$-coalgebra on $\S^n$ the unique coalgebra morphism $\S^n \to
  \S^{\Sigma^*}$ assigns to every element $i$ of $\S^n$ (perceived as the
  row input vector of the weighted automaton associated to the given
  coalgebra) the weighted language accepted by that automaton. 

\item An important special case of $\S$-weighted automata are ordinary
  non-deterministic automata. One takes $\S = \{0,1\}$ the Boolean
  semiring for which the category of $\S$-semimodules is (isomorphic
  to) the category of join-semilattices. Then
  $FX = \{0,1\} \times X^\Sigma$ is the coalgebraic type functor of
  deterministic automata with input alphabet $\Sigma$, and there is a
  bijective correspondence between an $F$-coalgebra on a free
  join-semilattice and non-deterministic automata. In fact in one
  direction one restricts $\powf X \to \{0,1\} \times (\powf X)^\Sigma$ to
  the set $X$ of generators, and in the other direction one performs
  the well-known subset construction. The final coalgebra is
  carried by the set of all formal languages on $\Sigma$ in this case.

\item
  Another special case is where $\S$ is a field. In this case,
  $\S$-semimodules are precisely the vector spaces over the field
  $\S$. Moreover, since every field is freely generated by its basis,
  it follows that the $\S$-weighted automata are precisely those
  $F$-coalgebras whose carrier is a finite dimensional vector space
  over $\S$.
\end{enumerate}
\end{exa}

We will now recall a few properties of algebraic categories $\Set^T$,
where $T$ is a finitary set monad, needed for our proofs.
\begin{rem}\label{rem:proj}
  \begin{enumerate}
  \item\label{rem:proj1} Recall that every strong epimorphism $e$ in $\Set^T$ is
    \emph{regular}, i.e.~$e$ is the coequalizer of some pair of $T$-algebra
    morphisms. It follows that the classes of strong and regular
    epimorphisms coincide, and these are precisely the surjective
    $T$-algebra morphisms. Similarly, jointly strongly epimorphic
    families of morphisms are precisely the jointly surjective
    families. Finally, monomorphisms in $\Set^T$ are precisely the
    injective $T$-algebra morphisms since the canonical forgetful
    functor $\Set^T \to \Set$ creates all limits (and pullbacks in
    particular). 
  \item\label{rem:proj2} Every free $T$-algebra $TX$ is \emph{(regular)
      projective}, i.e.~given any surjective $T$-algebra morphism $q:
    A \epito B$ then for every $T$-algebra morphism $h: TX \to B$ there
    exists a $T$-algebra morphism $g: TX \to A$ such that $q \cdot g =
    h$\iffull\/:
    \[
      \xymatrix{
        & A \ar@{->>}[d]^q \\
        TX \ar@{-->}[ru]^-g \ar[r]_h & B
        }
    \]\else\/.\fi
  \item\label{rem:proj3} Furthermore, note that every finitely
    presentable $T$-algebra $A$ is a regular ($=$ strong) quotient of
    a free $T$-algebra $TX$ with a finite set $X$ of
    generators. Indeed, $A$ is presented by finitely many generators
    and relations. So by taking $X$ as a finite set of generators of
    $A$, the unique extension of the embedding $X \subto A$ yields a
    surjective $T$-algebra morphism $TX \epito A$.
  \end{enumerate}
\end{rem}

\subsection{The Rational Fixed Point}

\takeout{% Stoffsammlung
\begin{itemize}
\item finitely presentable and finitely generated objects
\item lfp categories; examples
\item finitary functors; examples
\end{itemize}}% end takeout

As we mentioned in the introduction the canonical domain of behaviour
of `finite' coalgebras is the rational fixed point of an
endofunctor. Its theory can be developed for every finitary
endofunctor on a locally finitely presentable category. We will now
recall the necessary background material.

A \emph{filtered colimit} is the colimit of a diagram $\D \to \C$
where $\D$ is a filtered category (i.e.~every finite subcategory
$\D_0 \subto \D$ has a cocone in $\D$), and a \emph{directed colimit}
is a colimit whose diagram scheme $\D$ is a directed poset. A functor
is called \emph{finitary} if it preserves filtered (equivalently
directed) colimits. An object $C$ is called \emph{finitely
  presentable} (fp) if the hom-functor $\C(C, -)$ preserves filtered
(equivalently directed) colimits, and \emph{finitely generated} (fg)
if $\C(C, -)$ preserves directed colimits of monos (i.e.~colimits of
directed diagrams $D: \D \to \C$ where all connecting morphisms $Df$
are monic in $\C$). Clearly, every fp object is fg, but the converse
fails in general. In addition, fg objects are closed under strong epis
(quotients), which fails for fp objects in general.

A cocomplete category $\C$ is called \emph{locally finitely presentable}
(lfp) if there is a set of finitely presentable objects in $\C$ such that every object of $\C$ is a filtered colimit of objects from that set. We
refer to~\cite{ar} for further details.

\begin{exa}
  Examples of lfp categories are the categories of sets,
  posets and graphs, with finitely presentable objects precisely the
  finite sets, posets, and graphs, respectively.  The category of
  vector spaces over the field $k$ is lfp with finite-dimensional
  spaces being the fp-objects. Every algebraic category is lfp. The
  finitely generated objects are precisely the finitely generated
  algebras (in the sense of general algebra), and finitely presentable
  objects are precisely those algebras specified by finitely many
  generators and finitely many relations.
\end{exa}

\begin{exa}
  \label{ex:finfunc}
  Finitary functors abound. We just mention a few
  examples of finitary functors on $\Set$.

  Constant functors and the identity functor are, of course,
  finitary. For every finitary signature $\Sigma$,
  i.e.~$\Sigma = (\Sigma_n)_{n < \omega}$ is a sequence of sets with
  $\Sigma_n$ containing operation symbols of the finite arity $n$, the
  associated polynomial functor given by
  \[
    F_\Sigma X = \coprod_{n< \omega} \Sigma_n \times X^n
  \]
  is finitary. The finite power-set functor given by
  $\powf X = \{ Y \mid Y \subseteq X, \text{Y finite}\}$ is finitary
  (while the full power-set functor is not) and so is the bag functor
  mapping a set $X$ to the set of finite multisets on $X$. The class
  of finitary functors enjoys good closure properties: it is closed
  under composition, finite products, arbitrary coproducts, and, in
  fact, arbitrary colimits. As we have mentioned already, the finitary
  monads (i.e. whose functor part is finitary) on $\Set$ are precisely
  those monads whose Eilenberg-Moore category $\Set^T$ is (isomorphic
  to) a finitary variety of algebras.
\end{exa}

\begin{asms}
  For the rest of this section we assume that $F$ denotes a finitary
  endofunctor on the lfp category $\A$.
\end{asms}
\begin{rem}\label{rem:lfpprop}
  Recall that an lfp category, besides being cocomplete, is complete
  and has (strong epi, mono)-factorizations of morphisms~\cite{ar},
  i.e.~every morphism $f: X \to Y$ can be decomposed as
  $f = m \cdot e$ where $e: X \epito I$ is a strong epi and
  $m: I \monoto Y$ a mono. One should think of $I$ as the image of $X$
  in $Y$ under $f$.
\end{rem}

The rational fixed point is a fully abstract model of behaviour for
all $F$-coalgebras whose carrier is an fp-object. We now recall
its construction~\cite{amv_atwork}. 
\begin{nota}\label{nota:rho}
  Denote by $\coaf F$ the full subcategory of all $F$-coalgebras on fp
  carriers, and let $(\rho F, r)$ be the colimit of the inclusion
  functor of $\coaf F$ into $\coa F$:
  \iffull\[\else\/$\fi (\rho F, r) = \colim(\coaf F \subto \coa F)
    \iffull\]\else\/$ \fi with the colimit injections
  $\ha a: A \to \rho F$ for every coalgebra $a: A \to FA$ in
  $\coaf F$.

  We call $(\rho F, r)$ the \emph{rational fixed point} of $F$;
  indeed, it is a fixed point:
\end{nota}
%
%That this name is justified follows from the following
%
\begin{prop}[$\!\!$\cite{amv_atwork}]
  The coalgebra structure $r: \rho F \to F(\rho F)$ is an isomorphism.
\end{prop}

The rational fixed point can be characterized by a universal property
both as a coalgebra and as an algebra for $F$: as a coalgebra
$\rho F$ is the \emph{final locally finitely presentable
coalgebra}~\cite{m_linexp}, and as an algebra it is the \emph{initial
iterative algebra}~\cite{amv_atwork}. We will not recall the latter
notion as it is not needed for the technical development in this
paper.
Locally finitely presentable (locally fp, for short) coalgebras for $F$ can be
characterized as precisely those $F$-coalgebra obtained as a filtered
colimit of a diagram of coalgebras from $\coaf F$:
\begin{prop}[$\!\!$\cite{m_linexp}, Corollary~III.13]
  An $F$-coalgebra is locally fp if and only if it is a colimit of some
  filtered diagram $\D \to \coaf F \subto \coa F$. 
\end{prop}
For $\A = \Set$ an $F$-coalgebra $(X,c)$ is locally fp iff it is
\emph{locally finite}, i.e.~every element of $X$ is contained in a
finite subcoalgebra. Analogously, for $\A$ the category of vector
spaces over the field $k$ an $F$-coalgebra $(X,c)$ is locally fp iff
it is \emph{locally finite dimensional}, i.e.~every element of $X$ is
contained in a finite dimensional subcoalgebra.

Of course, there is a unique coalgebra morphism $\rho F \to \nu
F$. Moreover, in many cases $\rho F$ is \emph{fully abstract} for
locally fp
coalgebras, i.e.~besides being the final locally fp coalgebra the above
coalgebra morphism is monic; more precisely, if the classes of fp- and
fg-objects coincide and $F$ preserves non-empty monos, then $\rho F$
is fully abstract (cf.~Theorem~\ref{thm:fp=fg} below). The assumption
that the two object classes coincide is often true:

\begin{exa} \label{ex:fp=fg}
  \begin{enumerate}
  \item In the category of sets, posets, and graphs, fg-objects are fp
    and those are precisely the finite sets, posets, and graphs,
    respectively. 
  \item A \emph{locally finite variety} is a variety of algebras,
    where every free algebra on a finite set of generators is
    finite. It follows that fp- and fg-objects coincide and are
    precisely the finite algebras. Concrete examples are the
    categories of Boolean algebras, distributive lattices and
    join-semilattices.
  \item In the category of $\S$-semimodules for a semiring $\S$ the fp-
    and fg-objects need not coincide in general. However, if the
    semiring $\S$ is \emph{Noetherian} in the sense of \'Esik and
    Maletti~\cite{em_2011}, i.e.~every subsemimodule of a finitely
    generated $\S$-semimodule is itself finitely generated, then fg-
    and fp-semimodules coincide. Examples of Noetherian semirings are:
    every finite semiring, every field, every principal ideal domain
    such as the ring of integers and therefore every finitely
    generated commutative ring by Hilbert's Basis Theorem. The tropical
    semiring $(\Nat \cup \{\infty\}, \min, +, \infty, 0)$ is not
    Noetherian~\cite{em_2010}. The usual semiring of natural numbers
    is also not Noetherian: the $\Nat$-semimodule $\Nat \times \Nat$
    is finitely generated but its subsemimodule generated by the
    infinite set $\{(n,n+1) \mid n \geq 1\}$ is not. However,
    $\Nat$-semimodules are precisely the commutative monoids, and for
    them fg- and fp-objects coincide (this is known as Redei's
    theorem~\cite{redei}; see Freyd~\cite{freyd68} for a very short proof).
  \item\label{ex:fp=fg4} The category $\PCA$ of \emph{positively convex algebras}~\cite{Doberkat06,Doberkat08} is the
    Eilenberg-Moore category for the monad $\mathcal D$ of finitely supported
    subprobability distributions on sets. This monad maps a set $X$ to 
    \[
      \mathcal D X = \{d: X \to [0,1] \mid \text{$\supp d$ is finite and
        $\sum\limits_{x \in X} d(x) \leq 1$} \},
    \]
    where $\supp d = \{x \in X \mid d(x) \neq 0\}$, and a function
    $f: X \to Y$ to $\mathcal D f: \mathcal D X \to \mathcal D Y$ with
    \[
      \mathcal D f(d) = \lambda y. \sum\limits_{f x = y} d(x).
    \]
    More concretely, a positively convex algebra is a set $X$
    equipped with finite convex sum operations: for every $n$ and $p_1,
    \ldots, p_n \in [0,1]$ with $\sum\limits_{i = 1}^n p_i \leq 1$ we
    have an $n$-ary operation assigning to $x_1, \ldots, x_n \in X$
    an element $\bigboxplus\limits_{i=1}^n p_i x_i$ subject to the following
    axioms:
    \begin{enumerate}
    \item $\bigboxplus\limits_{i=1}^n p_i^kx_i = x_k$ whenever $p_k^k = 1$ and $p_i^k = 0$ for $i \neq k$, and
    \item $\bigboxplus\limits_{i=1}^n p_i \left(\bigboxplus\limits_{j=1}^k q_{i,j} x_j\right) = \bigboxplus\limits_{j=1}^k \left(\sum\limits_{i = 1}^n p_iq_{i,j}\right) x_j$.
    \end{enumerate}
    For $n = 1$ we write the convex sum operation for $p \in [0,1]$
    simply as $px$. The morphisms of $\PCA$ are maps preserving finite convex sums in the obvious sense. 

    The point of mentioning this example at length is that $\PCA$ is
    used for the coalgebraic modelling of the trace semantics of
    probabilistic systems (see e.g.~\cite{SilvaS11}), and recently, it
    was established by Sokolova and Woracek~\cite{SokolovaW15} that in
    $\PCA$, the classes of fp- and fg-objects coincide. We shall come
    back to this example in Section~\ref{sec:prop} when we introduce
    and discuss proper functors.
  \end{enumerate} 
\end{exa}

\begin{exa}\label{ex:rat}
  We list a number of examples of rational fixed points for cases
  where they do form subcoalgebras of the final coalgebra.
  \begin{enumerate}
  \item For the functor $FX = \{0,1\} \times X^A$ on $\Set$ the
    finite coalgebras are deterministic automata, and the rational fixed
    point is carried by the set of regular languages on the alphabet
    $A$.
  \item For every finitary signature $\Sigma$, the final coalgebra for
    the associated polynomial functor $F_\Sigma$ (see Example~\ref{ex:finfunc})
    is carried by the set of all (finite and infinite) $\Sigma$-trees,
    i.e.~rooted and ordered trees where each node with $n$-children is
    labelled by an $n$-ary operation symbol. The rational fixed point
    is the subcoalgebra given by \emph{rational} (or
    \emph{regular}~\cite{courcelle}) $\Sigma$-trees, i.e.~those
    $\Sigma$-trees that have only finitely many different subtrees (up
    to isomorphism) -- this characterization is due to
    Ginali~\cite{ginali}. For example, for the signature $\Sigma$ with
    a binary operation symbol $*$ and a constant $c$ the following
    infinite $\Sigma$-tree (here written as an infinite term) is
    rational:
    \[
      c * (c * (c* \cdots )));
    \]
    in fact, its only subtrees are the whole tree and the single node
    tree labelled by $c$.
  \item For the functor $FX = \Real \times X$ on $\Set$ the final
    coalgebra is carried by the set $\Real^\omega$ of real streams,
    and the rational fixed point is carried by its subset of
    eventually periodic streams (or lassos). Considered as a functor
    on the category of vector spaces over $\Real$, the final
    coalgebra $\nu F$ remains the same, but the rational fixed point
    $\rho F$ consists of all rational streams~\cite{rutten_rat}.
  \item For the functor $FX = \S \times X^A$ on the category $\smod$
    of $\S$-semimodules for the semiring $\S$ we already mentioned
    that $\nu F = \S^{A^*}$ consists of all formal
    power-series. Whenever the classes of fg- and fp-semimodules
    coincide, e.g.~for every Noetherian semiring $\S$ or the semiring
    of natural numbers, then $\rho F$ is formed by the \emph{recognizable}
    formal power-series; from the Kleene-Schützenberger
    theorem~\cite{schuetzenberger} (see also~\cite{br88}) it follows
    that these are, equivalently, the \emph{rational} formal
    power-series.

  \item On the category of presheaves $\Set^{\mathcal F}$, where
    $\mathcal F$ is the category of all finite sets and maps between
    them, consider the functor $FX = V + X \times X + \delta(X)$, where
    $V:  \mathcal F \subto \Set$ is the embedding and
    $\delta (X)(n) =X(n+1)$. This is a paradigmatic example of a
    functor arising from a \emph{binding signature} for which initial
    semantics was studied by Fiore et al.~\cite{fpt}.

    The final coalgebra $\nu F$ is carried by the presheaf of all
    $\lambda$-trees modulo $\alpha$-equivalence: $\nu F(n)$ is the set
    of (finite and infinite) $\lambda$-trees in $n$ free variables
    (note that such a tree may have infinitely many bound
    variables). And $\rho F$ is carried by the rational
    $\lambda$-trees, where an $\alpha$-equivalence class is called
    \emph{rational} if it contains at least one $\lambda$-tree which
    has (up to isomorphism) only finitely many different subtrees
    (see~\cite{amv_horps_full} for details). Rational $\lambda$-trees
    also appear as the rational fixed point of a very similar functor
    on the category of nominal sets~\cite{mw15}. An analogous
    characterization can be given for every functor on nominal sets
    arising from a binding signature~\cite{msw16}.
  \end{enumerate}
\end{exa}

As we mentioned previously, whether fg- and fp-objects coincide is
currently unknown in some base categories used in the coalgebraic
modelling of systems, for example, in idempotent semirings (used in
the treatment of context-free grammars \cite{wbr13})\smnote{Do we have
  a counterexample or not? Check with Joost Winter.}, in
algebras for the stack monad (used for modelling configurations of
stack machines~\cite{gms14}); or it even fails, for example in
the category of finitary monads on sets (used in the categorical study
of algebraic trees \cite{amv_secalg}) or in Eilenberg-Moore
categories for a monad in general (the target categories of generalized
determinization \cite{sbbr13}).

As a remedy, in recent joint work with Pattinson and
Wi\ss\/mann~\cite{mpw17}, we have introduced the \emph{locally finite
  fixed point} which provides a fully abstract model of finitely
generated behaviour. Its construction is very similar to that of the
rational fixed point but based on fg- in lieu of fp-objects. In more
detail, one considers the full subcategory $\coafg F$ of all $F$-coalgebras
carried by an fg-object and takes the colimit of its inclusion
functor:
\[
  (\theta F, \ell) = \colim(\coafg F \subto \coa F).
\]
\begin{thm}[\!\!\cite{mpw17}, Theorems~3.10 and 3.12]\label{thm:lffsub}
  Suppose that the finitary functor $F: \A \to\A$ preserves
  non-empty monos. Then $(\theta F, \ell)$ is a fixed point for $F$, and it is a
  subcoalgebra of $\nu F$. 
\end{thm}
\begin{rem}\label{rem:fs}
  \begin{enumerate}
  \item Note that for an arbitrary (not necessarily concrete) lfp category
    $\A$ the notion of a non-empty monomorphisms needs explanation: a
    monomorphism $m: X \monoto Y$ is said to be \emph{empty} if its
    domain $X$ is a strict initial object of $\A$, where recall that the
    initial object $0$ of $\A$ is \emph{strict} provided that every
    morphism $A \to 0$ is an isomorphism.  In particular, if the initial
    object of $\A$ is not strict, then all monomorphisms are
    non-empty.
  \item\label{i:fs2} For a functor $F: \A \to \A$ preserving non-empty monos the
    category $\coa F$ of all $F$-coalgebras inherits the (strong epi,
    mono)-factorization system from $\A$ (see Remark~\ref{rem:lfpprop})
    in the following sense: every coalgebra morphism
    $f: (X,c) \to (Y,d)$ can be factorized into coalgebra morphisms
    $e$ and $m$ carried by a strong epi and a mono in $\A$,
    respectively. In fact, one (strong epi, mono)-factorizes
    $f = m \cdot e$ in $\A$ and obtains a unique coalgebra
    structure on the `image' $I$ such that $e$ and $m$ are coalgebra
    morphisms:
    \[
      \xymatrix{
        X \ar[r]^-c \ar@{->>}[d]_e & FX \ar[d]^-{Fe} \\
        I \ar@{-->}[r] \ar@{ >->}[d]_m & FI \ar@{ >->}[d]^{Fm} \\
        Y \ar[r]_-d & FY
      }
    \]
    Indeed, if $m$ is a non-empty mono, we know that $Fm$ is monic by
    assumption and we use the unique diagonal property. Otherwise, $m$
    is an empty mono, which implies that $e: X \epito I$ is an
    isomorphism since $I$ is a strict initial object. Then $Fe \cdot c
    \cdot e^{-1}$ is the desired coalgebra structure on $I$.
  \end{enumerate}
\end{rem}
Furthermore, like its brother, the rational fixed point, $\theta F$
is characterized by a universal property both as a coalgebra and as
an algebra: it is the final locally finitely generated coalgebra and
the initial fg-iterative algebra~\cite[Theorem~3.8 and Corollary~4.7]{mpw17}. 

Under additional assumptions, which all hold in every algebraic
category, we have a close relation between $\rho F$ and $\theta F$; in
fact, the following is a consequence of~\cite[Theorem~5.4]{mpw17}:
\begin{thm}
  \label{thm:rhotheta}
  Suppose that $\A$ is an lfp category such that every fp-object is a
  strong quotient of a strong epi projective fp-object, and let $F: \A
  \to \A$ be finitary and preserving non-empty monos.  
  Then $\theta F$ is the image of $\rho F$ in the final coalgebra.
\end{thm}
More precisely, taking the (strong epi, mono)-factorization of the
unique $F$-coalgebra morphism $\rho F \to \nu F$ yields $\theta F$,
i.e.~for $F$ preserving monos on an algebraic category we have the
following picture:
\[
  \rho F \epito \theta F \monoto \nu F.
\]
A sufficient condition under which $\rho F$ and $\theta F$ coincide
is the following (cf.~\cite[Section~5]{mpw17}):
\begin{thm}\label{thm:fp=fg}
  Suppose that in addition to the assumption in
  Theorem~\ref{thm:rhotheta} the classes of fg- and fp-objects
  coincide in $\A$. Then $\rho F \cong \theta F$, i.e.~the left-hand
  morphism above is an isomorphism.
\end{thm}

In the introduction we briefly mentioned a number of interesting
instances of $\theta F$ that are not (known to be) instances of the
rational fixed point; see~\cite{mpw17} for details.  
A concrete example, where $\rho F$ is not a subcoalgebra of $\nu F$ (and
hence not isomorphic to $\theta F$) was given
in~\cite[Example~3.15]{bms13}. We present a new, simpler example based on
similar ideas:
\begin{exa}\label{ex:nonex}
  \begin{enumerate}
  \item Let $\A$ be the category of algebras for the signature
    $\Sigma$ with two unary operation symbols $u$ and $v$. The natural
    numbers $\Nat$ with the successor function as both operations
    $u^\Nat$ and $v^\Nat$ form an object of $\A$. We consider the
    functor $FX = \Nat \times X$ on $\A$. Coalgebras for $F$ are
    automata carried by an algebra $A$ in $\A$ equipped with two
    $\Sigma$-algebra morphisms: an output morphism $A \to \Nat$ and a
    next state morphism $A \to A$. The final coalgebra is carried by 
    the set $\Nat^\omega$ of streams of natural numbers with the coordinatewise
    algebra operations and with the coalgebra structure given by the usual
    head and tail functions.

    Note that the free $\Sigma$-algebra on a set $X$ of generators is
    $TX \cong \{u,v\}^* \times X$; we denote its elements by $w(x)$
    for $w \in \{u,v\}^*$ and $x \in X$.  The operations are given by prefixing
    words by the letters $u$ and $v$, respectively: $s^{TX}: w(x) \mapsto sw(x)$
    for $s = u$ or $v$. 

    Now one considers the $F$-coalgebra $a: A \to FA$, where
    $A = T\{x\}$ is free $\Sigma$-algebra on one generator $x$ and $a$
    is determined by $a(x) = (0, u(x))$. Recall our notation
    $\fin a: A \to \nu F$ for the unique coalgebra morphism. Clearly,
    $\fin a(x)$ is the stream $(0,1,2,3,\cdots)$ of all natural
    numbers, and since $\fin a$ is a $\Sigma$-algebra morphism we have
    \[
      \fin a (u(x)) = \fin a (v(x)) = (1,2,3,4,\cdots).
    \]
    Since $A$ is (free) finitely generated, it is of course, finitely
    presentable as well. Thus, $(A,a)$ is a coalgebra in $\coaf
    F$.
    \iffull
    However, we shall now prove that the (unique) $F$-coalgebra morphism
    $\ha a: A \to \rho F$ maps $u(x)$ and $v(x)$ to two distinct
    elements of $\rho F$. 

    We prove this by contradiction. So suppose that
    $\ha a(u(x)) = \ha a(v(x))$. By the construction of $\rho F$ as a
    filtered colimit (see Notation~\ref{nota:rho}) we know that there
    exists a coalgebra $b: B \to FB$ in $\coaf F$ and an $F$-coalgebra
    morphism $h: A \to B$ with
    \begin{equation}\label{eq:h}
      h(u(x)) = h(v(x)).
    \end{equation}
    Since $B$ is a finitely presented $\Sigma$-algebra it is the
    quotient in $\A$ of a free algebra $A'$ via some surjective
    $\Sigma$-algebra morphism $q: A' \epito B$, say. Next observe,
    that there is a coalgebra structure $a': A' \to FA'$ such that $q$
    is an $F$-coalgebra morphism from $(A', a')$ to $(B,b)$: for $Fq$
    is a surjective $\Sigma$-algebra morphism and so we obtain $q'$ by
    using projectivity of $A'$ w.r.t.~$b \cdot q: A' \to FB$
    (cf.~Remark~\ref{rem:proj}(\ref{rem:proj2})):
    \[
      \xymatrix{
        A' \ar@{-->}[r]^{a'} \ar@{->>}[d]_q & FA'\ar@{->>}[d]^{Fq} \\
        B \ar[r]_-b & FB
        }
    \]

    Now choose a term $t_x$ in $A'$ with $q(t_x) = h(x)$. Using that
    $q$ and $h$ are $\Sigma$-algebra morphisms we see that $q(u(t_x))
    = q(v(t_x))$ as follows:
    \begin{equation}\label{eq:q}
      q(u(t_x)) 
      =
      u^B(q(t_x))
      =
      u^B(h(x))
      =
      h(u(x))
      =
      v^B(h(x))
      =
      v^B(q(t_x))
      =
      q(v(t_x)).
    \end{equation}
    Since $h$ is an $F$-coalgebra morphism, we obtain
    from~\eqref{eq:h} that $h$ merges the right-hand components of
    $a(u(x))$ and $a(v(x))$, in symbols: $h(uu(x)) = h(vu(x))$. It
    follows that $q$ satisfies $q(uu(t_x)) = q(vu(t_x))$ using a
    similar argument as in~\eqref{eq:q} above.

    Continuing to use that $h$ and $q$ are $F$-coalgebra morphisms, we
    obtain the following infinite list of elements (terms) of $A'$
    that are merged by $q$ (we write these pairs as equations):
    \begin{equation}
      \label{eq:rels}
      q(u^{n+1}(t_x)) = q(vu^n(t_x)) \qquad \text{for $n \in \Nat$.}
    \end{equation}
    
    We need to prove that there exists no finite set of relations
    $E \subseteq A' \times A'$ generating the above congruence on $A'$
    given by $q: A' \epito B$. So suppose the contrary, and let $A_0'$
    be the $\Sigma$-subalgebra of $A'$ generated by $\{t_x\}$,
    i.e.~$A_0' \cong \{u,v\}^* \times \{t_x\}$. Since $q(t_x) = h(x)$
    and $q$ and $h$ are both coalgebra morphisms we know that
    $\fin a' = \fin b \cdot q$ and $\fin b \cdot h = \fin a$ and
    therefore
    \[
      \fin a' (t_x) = \fin b(q(t_x)) = \fin b (h(x)) = \fin a(x) =
      (0,1,2,3, \cdots).
    \]
    Since $\fin a'$ is a $\Sigma$-algebra morphism it follows that for
    a word $w \in \{u,v\}^*$ of length $n$ we have
    \begin{equation}\label{eq:cp}
      \fin a'(w(t_x)) = (n, n+1, n+2, n+3, \cdots).
    \end{equation}
    Thus, when $w, w' \in \{u,v\}^*$ are of different length, then the pair 
    $(w(t_x), w'(t_x))$ cannot be in the congruence generated by $E$;
    otherwise we would have $q(w(t_x)) = q(w'(t_x))$ which implies
    $\fin a'(w(t_x)) = \fin a'(w'(t_x))$ contradicting~\eqref{eq:q}.

    Now let $\ell$ be the maximum length of words from $\{u,v\}^*$
    occurring in any pair contained in the finite set $E$. Then the pair
    $(u^{\ell+2}(t_x), vu^{\ell+1}(t_x))$ obtained from the
    $\ell+1$-st equation in~\eqref{eq:rels} is not in the congruence
    generated by $E$; for if any pair of terms of height greater then
    $\ell$ are related by that congruence, these two
    terms must have the same head symbol. Thus we arrive at a
    contradiction as desired. 
    \else
    However, one can prove that the (unique) $F$-coalgebra morphism 
    $\ha a: A \to \rho F$ satisfies $\ha a(u(x)) \neq \ha a(v(x))$,
    see the full paper for details~\cite{milius17}.
    \fi
 
  \item In this example we also have that $\theta F$ and $\nu F$ do
    not coincide. To see this we use that $\theta F$ is the union of
    images of all $\fin c: TX \to \nu F$ where $(TX,c)$ ranges over
    those $F$-coalgebras whose carrier $TX$ is free finitely generated
    (i.e.~$TX \cong \{u,v\}^* \times X$ for some finite set $X$)
    \cite[Theorem~6.5]{mpw17}. Hence, each such algebra $TX$ is
    countable, and there exist only countably many of them, up to
    isomorphism. Furthermore, note that on every free finitely
    generated algebra $TX$ there exist only countably many coalgebra
    structures $c: TX \to FTX$, since $FTX = \Nat \times TX$ is
    countable and $c$, being a $\Sigma$-algebra morphism, is
    determined by its action on the finitely many generators.
    Thus, $\theta F$ is countable because it is the above union of
    countably many countable coalgebras. However, $\nu F$ being carried
    by the set $\Nat^\omega$ of all streams over $\Nat$ is uncountable. 

    \takeout{% simpler argument using cardinality above
  Note that being a $\Sigma$-algebra morphism any coalgebra structure
  $a: TX \to FTX$ is determined by its action on the generators. And
  from the form of any $TX$ we know that for any $x \in X$ there
  exist $k,n_i \in \Nat$, $w_i \in \{u,v\}^*$
  and $x_i \in X$, $i = 1, \ldots k$, such that $x = x_0$ and
  \begin{align*}
    a(x_i) & = (n_i, w_i(x_{i+1})) & \text{ for $i = 0, \ldots, k-1$
                                     and} \\
    a(x_k) & = (n_k, w_k(x_j)) &  \text{ for some $j \in \{0, \ldots, k\}$.}   
  \end{align*}
  Now let $m_i = |w_i|$, $i = 1, \ldots, k$, be the lengths of words. Then
  it follows that 
  \[
    \fin a (x_0) = (n_0, m_0 + n_1, m_0 + m_1 + n_2, 
    \cdots, m_0 + \cdots + m_{k-1} + n_k, m_0 + \cdots m_k +
    n_j,\cdots).
  \]
  Let $m$ be the maximum of all $n_i$ and $m_i$. Then it is clear that
  the $n$-th entry of $\fin a(x_0)$ can be at most $(n+1) \cdot m$. It
  follows that for any $w \in \{u,v\}^*$ the $n$-th entry of
  $\fin a(w(x))$ is bounded above by $(n+1)\cdot m + |w|$. Thus, 
  the entries of every stream in $\theta F$ grow at
  most linearly. However, there are streams in $\nu F$ for
  which this is not the case, e.g.~the stream $(1,2,4,8, \cdots)$ of
  powers of $2$. Hence $\theta F$ %and $\nu F$ are not isorphic. 
  does not coincide with $\nu F$.}% end takeout
\end{enumerate} 
\end{exa}

\section{A Fixed Point Based on Coalgebras Carried by Free Algebras}
\label{sec:phi}

In this section we study coalgebras for a functor $F$ on an
algebraic category $\Set^T$ whose carrier is a free finitely generated
algebra. These coalgebras are of interest because they are precisely
those coalgebras arising as the results of the generalized
determinization~\cite{sbbr13}.

We shall see that their colimit yields yet another fixed point of $F$
(besides the rational fixed point and the locally finite
one). Moreover, in the next section we show that this fixed point is
characterized by a universal property as an algebra.
 
\takeout{
The purpose of this section is to study the situation where the
rational fixed point for a functor $F$ on an algebraic category
$\Set^T$ coincides with the locally finite one, and moreover, both can
be constructed just from those coalgebras whose carrier is a free
finitely generated coalgebra. The latter coalgebras are precisely
those coalgebras arising as the results of the generalized
determinization~\cite{sbbr13}.
}% end takeout

\begin{asms}\label{ass:ass}
  Throughout the rest of the paper we assume that $\A$ is an
  \emph{algebraic category}, i.e.~$\A$ is (equivalent to) the
  Eilenberg-Moore category $\Set^T$ for a finitary monad $T$ on $\Set$.
  Furthermore, we assume that $F: \A \to \A$ is a finitary
  endofunctor preserving surjective $T$-algebra morphisms.
  \smnote[inline]{Do I need that $F$ preserves non-empty monos?}
\end{asms}

\begin{rem}\label{rem:bialg}
  \begin{enumerate}
  \item Note that we do not assume here that $F$ preserves non-empty
    monomorphisms (cf.~Theorems~\ref{thm:lffsub} and~\ref{thm:rhotheta}) as this
    assumption is not needed for our main result
    Theorem~\ref{thm:main}. However, we will make this assumption at
    the end, in order to obtain the picture in~\eqref{diag:hierarchy} (see
    Corollary~\ref{cor:fp=fg}). 

  \item The most common instance of a functor $F$ on an algebraic category
    $\A$ is a lifting of an endofunctor $F_0: \Set \to \Set$, i.e.~we
    have a commutative square 
    \[
      \xymatrix{
        \Set^T \ar[r]^-F \ar[d]_U & \Set^T \ar[d]^{U} \\
        \Set \ar[r]_-{F_0} & \Set
      }
    \]
    where $U: \A \to \Set$ is the forgetful functor. Recall that
    monomorphisms in $\Set^T$ are precisely the injective $T$-algebra
    morphisms (see Remark~\ref{rem:proj}(\ref{rem:proj1})). Hence, a
    lifting $F$ preserves all non-empty monos since the lifted set
    functor $F_0$ does so. Similarly, $F$ preserves surjective $T$-algebra
    morphisms since $F_0$ preserves surjections (which
    are split epis in $\Set$). Finally, $F$ is finitary whenever $F_0$
    is so because filtered colimits in $\Set^T$ are created by $U$.

  \item It is well known that liftings $F: \Set^T \to \Set^T$ are in
    bijective correspondence with distributive laws of the monad $T$
    over the functor $F_0$, i.e.~natural transformations
    $\lambda: TF_0 \to F_0T$ satisfying two obvious axioms w.r.t.~the
    unit and multiplication of $T$ (see
    e.g.~Johnstone~\cite{johnstone_lift}):
    \[
      \xymatrix{
        F_0 \ar[r]^-{F_0\eta} \ar[rd]_-{\eta F_0} & TF_0
        \ar[d]^{\lambda}
        \\
        & F_0T
      }
      \qquad
      \xymatrix{
        TTF_0
        \ar[r]^-{T\lambda}
        \ar[d]_{\mu F_0}
        & TF_0T \ar[r]^-{\lambda T} &
        F_0TT \ar[d]^{F_0 \mu} \\
        TF_0 \ar[rr]_-{\lambda} && F_0T
        }
    \]
    Moreover, coalgebras for the lifting $F$ are precisely the
    \emph{$\lambda$-bialgebras}, i.e.~sets $X$ equipped with an
    Eilenberg-Moore algebra structure $\alpha: TX \to X$ and a coalgebra
    structure $c: X \to F_0X$ subject to the following commutativity
    condition
    \[
      \xymatrix{
        TX \ar[d]_\alpha \ar[r]^-{Tc} 
        &  
        TF_0X \ar[r]^-{\lambda_X}  
        & 
        F_0TX \ar[d]^{F_0 \alpha} 
        \\
        X \ar[rr]_-c 
        & &
        F_0X
      }
    \]
    which states that $c$ is a $T$-algebra morphism from $(X,\alpha)$
    to $F(X,\alpha)$.
    
\item Let $F_0: \Set \to \Set$ have a lifting to $\Set^T$. \emph{Generalized
    determinization}~\cite{sbbr13} is the process of turning a given
  coalgebra $c: X \to F_0TX$ in $\Set$ into the coalgebra
  $\kl c: TX \to FTX$ in $\Set^T$. For example,
  for the functor $F_0X = \{0,1\} \times X^\Sigma$ on $\Set$ and the
  finite power-set monad $T = \powf$, $F_0T$-coalgebras are precisely
  non-deterministic automata and generalized determinization is the
  construction of a deterministic automaton by the well-known subset
  construction. The unique $F$-coalgebra morphism $\fin{(\kl c)}$
  assigns to each state $x \in X$ the language accepted by $x$ in the
  given non-deterministic automaton (whereas the final semantics for
  $F_0T$ on $\Set$ provides a kind of process semantics taking the
  non-deterministic branching into account).

  Thus studying the behaviour of $F$-coalgebras whose carrier is a
  free finitely generated $T$-algebra $TX$ is precisely the study of a
  \emph{coalgebraic language semantics} of finite $F_0T$-coalgebras.
\end{enumerate}
\end{rem}

\begin{nota}
  We denote by 
  \iffull\[\else\/$\fi
    \coafr F
  \iffull\]\else\/$ \fi
  the full subcategory of $\coa F$ given by all coalgebras
  $c: TX \to FTX$ whose carrier is a free finitely generated
  $T$-algebra, i.e.~where $X$ is a finite set $X$.

  The colimit of the inclusion functor of $\coafr F$ into the category
  of all $F$-coalgebras is denoted by
  \iffull\[\else\/$\fi
    (\phi F, \zeta ) = \colim(\coafr F \subto \coa F)
  \iffull\]\else\/$ \fi
  with the colimit injections $\inj_c: TX \to \phi F$ for every $c: TX \to
  FTX$ in $\coafr F$. 
\end{nota}

\takeout{% this is wrong
\begin{rem}
  The coalgebra $\phi F$ is an lfp coalgebra. Indeed, as we have just
  seen $\phi F$ is the colimit of a filtered diagram of $F$-coalgebras with a
  finitely presentable carrier; indeed, each $TX$ with $X$ finite is
  finitely presentable, and finitely presentable objects are closed
  under coequalizers. 
\end{rem}}%end takeout
\begin{nota}\label{not:h}
  Since every free finitely generated algebra $TX$ is clearly fp
  (being presented by the finite set $X$ of generators and no
  relations), $\coafr F$ is a full subcategory of $\coaf F$. Therefore,
  the universal property of the colimit $\phi F$ induces a coalgebra
  morphism denoted by $h: \phi F \to \rho F$. Furthermore we write
  $m: \phi F \to \nu F$ for the unique $F$-coalgebra morphisms into
  the final coalgebra, respectively.
\end{nota}
We shall show in Proposition~\ref{prop:hepi} that $h$ is a strong
epimorphism. Thus, whenever $F$ preserves non-empty monos, we have the
picture~\eqref{diag:hierarchy} from the introduction.

\begin{rem}
  \label{rem:sifted}
  We will also use that the colimit $\phi F$ is a sifted colimit. 
  \begin{enumerate}
  \item Recall that a small category $\D$ is called
    \emph{sifted}~\cite{arv} if finite products commute with colimits
    over $\D$ in $\Set$. More precisely, $\D$ is sifted iff given any
    diagram $D: \D \times \J \to \Set$, where $\J$ is a finite discrete
    category, the canonical map
    \[
      \colim\limits_{d \in \D} \Big(\prod\limits_{j \in \J}
        D(d,j)\Big) 
      \to
      \prod\limits_{j \in \J} (\colim\limits_{d \in \D} D(d,j))
    \]
    is an isomorphism. A \emph{sifted colimit} is a colimit of a
    diagram with a sifted diagram scheme. 

  \item It is well-known that the forgetful functor $\Set^T \to \Set$
    preserves and reflects sifted colimits; this follows
    from~\cite[Corollary~11.9]{arv}.\smnote{Previously, I had
      in~\cite[Proposition~2.5]{arv} here.}

  \item\label{rem:sifted3} Further recall~\cite[Example~2.16]{arv} that every small
    category $\D$ with finite coproducts is sifted. Thus, from
    Lemma~\ref{lem:coprod} below it follows that $\D = \coafr F$ is
    sifted, and therefore $\phi F$ is a sifted colimit.
  \end{enumerate} 
\end{rem}
\begin{lem}
  \label{lem:coprod}
  The category $\coafr F$ is closed under finite coproducts in $\coa F$.  
\end{lem}
\begin{proof}
  The empty map $0 \to FT0$ extends uniquely to a $T$-algebra morphism
  $T0 \to FT0$, i.e.~an $F$-coalgebra, and this coalgebra is the
  initial object of $\coafr F$. 

  Given coalgebras $c: TX \to FTX$ and $d: TY \to FTY$ one uses that
  $T(X+Y)$ together with the injections $T\inl: TX \to T(X+Y)$ and
  $T\inr: TY \to T(X+Y)$ form a coproduct in $\Set^T$. This implies
  that forming the coproduct of $(TX, c)$ and $(TY, d)$ in $\coa F$ we
  obtain an $F$-coalgebra on $T(X+Y)$, and this is an object of
  $\coafr F$ since $X+Y$ is finite. 
\end{proof}
\begin{thm}[Urbat~\cite{Urbat17}, Lemma~4.5]\label{thm:phifix}
  If $F$ preserves sifted colimits, then $\phi F$ is a fixed point of
  $F$, i.e.~$\zeta: \phi F \to F(\phi F)$ is an isomorphism.
\end{thm}
\smnote{TODO: it could be a good idea to include an explicit proof of
  this fact here.}  
Recall that every finitary endofunctor on $\Set$ preserves sifted
colimits (this follows from~\cite[Corollary~6.30]{arv}). Thus, so does
every lifting $F:\Set^T \to \Set^T$ of a finitary endofunctor on
$\Set$, using Remark~\ref{rem:sifted}(2). In general, finitary
functors need not preserve sifted colimits~\cite[Example~7.11]{arv}.

One might now expect that $\phi F$ is characterized as a coalgebra by a
universal property similar to finality properties that characterize
$\rho F$ and $\theta F$. However, Urbat~\cite{Urbat17} shows that this
is not the case. In fact, he provides the following example of a
coalgebra $c: TX \to FTX$ where $\inj_c: TX \to \phi F$ is not the
only $F$-coalgebra morphism:
\begin{exa}\label{ex:nonproper}
  \begin{enumerate}
  \item Let $\A$ be the category of algebras for the signature with
    one unary operation symbol $u$ (and no equations), and let
    $F= \Id$ be the identity functor on $\A$. Let $A$ be the free
    (term) algebra on one generator $x$, and let $B$ be the free
    algebra on one generator $y$ (i.e.~both $A$ and $B$ are
    isomorphic to $\Nat$). We equip $A$ and $B$ with the $F$-coalgebra
    structures $a = \id: A \to A$ and $b: B \to B$ given by
    $b(y) = u(y)$. 
    \iffull%
    Then the mapping $t \mapsto u(t)$ clearly is an
    $F$-coalgebra morphism from $B$ to itself, i.e.~a morphism in
    $\coafr F$. Therefore we have $\inj_b (y) = \inj_b(u(y))$. 

    Now define a morphism $g: A \to \phi F$ in $\A$ by $g(x) =
    \inj_b(y)$. Then $g$ is an $F$-coalgebra morphism since
    \[
      g \o a(x) = g(x) = \inj_b(y) = \inj_b(u(y)) = \inj_b(b(y)) =
      \zeta(\inj_b(y)) = \zeta(g(x)),
    \]
    where $\zeta: \phi F \to F(\phi F)$ is the coalgebra structure on
    $\phi F$.
    
    We prove the following property: for every morphism $f$ in
    $\coafr F$ from $\alpha: TX \to TX$ to $\beta: TY \to TY$, any
    $t \in TX$ reaches finitely many states iff $f(t)$ does so, more
    precisely:
    \[
      \{ \alpha^n(t) \mid n \in \Nat\} \text{ is finite} \qquad \iff
      \qquad
      \{\beta^n(f(t)) \mid n \in \Nat\} \text{ is finite}. 
    \]
    Indeed, if $t$ reaches finitely many states, then the $f(\alpha^n(t))$,
    for $n \in \Nat$, form a finite set, and $\beta^n(f(t))$, $n \in
    \Nat$ is the same set since $f$ is a coalgebra morphism. 

    Conversely, suppose that $t$ reaches infinitely many states. Since
    $f$ is a morphism in $\A$, we know that if $\alpha^n(t)$ is
    $u^k(x)$ for some $x \in X$ then $f(\alpha^n(t)) = \beta^n(f(t))$
    must be $u^l(y)$ with $l \geq k$ for some $y \in Y$. Thus, $f(t)$
    must also reach infinitely many states.

    We can now conclude that $g, \inj_a: A \to \phi F$ are different
    coalgebra morphisms. Indeed, $\inj_a(x)$ reaches only itself since $x$
    does so, but $g(x) = \inj_b(y)$ reaches infinitely many states
    since $y$ does so. Thus, $g(x) \neq \inj_a(x)$. 
    
    It follows that $|\phi F| \geq 2$, while
    $\rho F = \theta F = \nu F = 1$; to see the latter equation use
    that $\id: 1 \to 1$ is a coalgebra in $\coaf F$ since $1$ is the
    object of $\A$ presented by one generator $z$ and one relation
    $z = u(z)$.
    \else
    Define $g: A \to \phi F$ by $g(x) = \inj_b(y)$. Then one can show
    that $g$ is an $F$-coalgebra morphism different from the
    $F$-coalgebra morphism $\inj_a: A \to \phi F$. 
    \fi

  \item Using similar ideas as in the previous point one can show
    that, for the category $\A$ and $FX = \Nat \times X$ from
    Example~\ref{ex:nonex}, $\phi F$ and $\rho F$ do not
    coincide\iffull\/. \else\/, see the full paper~\cite{milius17}. \fi 
    Consequently, in this example, none of the arrows
    in~\eqref{diag:hierarchy} is an isomorphism.

    \iffull
    In order to see that $\phi F$ and $\rho F$ do not coincide, consider
    the two coalgebras $a: A \to FA$ and $b:B \to FB$ with $A =
    T\{x\}$ and $B = T\{y\}$ and with the coalgebra structure given by
    $a(x) = (0,u(x))$ and $b(y) = (0,v(y))$. These coalgebras both
    lie in $\coafr F$. Consider also the coalgebra $p: P \to FP$ where
    $P$ is presented by one generator $z$ and one relation $u(z) =
    v(z)$, i.e.~$P = T\{z\}/\mathord{\sim}$, where $\sim$ is the
    smallest congruence with $u(z) \sim v(z)$. Hence, $w(z) \sim
    w'(z)$ for $w,w' \in \{u,v\}^*$ iff $w$ and $w'$ have the same
    length. The coalgebra structure is defined by $p([w(x)]) = (0,
    [uw(x)])$. The coalgebra $(P,p)$ lies in $\coaf F$. Now $f: A \to
    P$ and $g: B \to P$ determined by $f(x) = z = g(y)$ are easily
    seen to be $F$-coalgebra morphisms, and therefore $\ha a = \ha p
    \cdot f$ and $\ha b = \ha p \cdot g$. Therefore 
    \[
      \ha a(x) = \ha p (f(x)) = \ha p(z) = \ha p (g(y)) = \ha b (y).
    \]

    However, we will prove that $\inj_a (x) \neq \inj_b (x)$.  For any
    $(TX, c)$ in $\coafr F$ and $t \in TX$, we say that 
    \emph{$t$-reachable states are $u$-bounded} if there exists a natural
    number $k$ such that, for any state $s = w(x)$ reachable from $t$
    via the next state function, 
    the number $|w|_u$ of $u$'s in $w$ is at most $k$. Now we prove for any
    morphism $f: (TX, c) \to (TY, d)$ in $\coafr F$ and
    any $t \in TX$ the following claim:
    \[
      \text{$t$-reachable states are $u$-bounded iff
        $f(t)$-reachable states are $u$-bounded.}
    \]
    Indeed, a state $s = w(x)$ is reachable from $t$ iff $f(s) = wf(x)$
    is reachable from $f(t)$.  
    Then the 'only if' direction is clear: if $t$-reachable states are not
    $u$-bounded, then neither are $f(t)$-reachable states. 
    For the 'if' direction suppose $t$-reachable states are
    $u$-bounded by $k$. Then $f(t)$-reachable states are bounded by $k
    + \max\{|f(x)|_u \mid x \in X\}$. 
    \fi
  \end{enumerate}
\end{exa}

\takeout{
In this section we are going to investigate when the first three fixed
points in~\eqref{diag:hierarchy} collaps to one, i.e.~$\phi F \cong \rho F
\cong \theta F$. As a consequence, it follows
that finality of a given locally fp coalgebra for $F$ can be established by
checking the universal property only for the coalgebras in $\coafr F$ 
(Corollary~\ref{cor:final}).   
}% end takeout

\noindent
Coming back to the discussion of properties that $\phi F$ does
have, the following proposition shows that $\rho F$ is always a strong
quotient of $\phi F$. Recall from  Notation~\ref{not:h} the canonical coalgebra
morphism $h$ from $\phi F$ to $\rho F$:
\begin{prop}
  \label{prop:hepi}
  The morphism $h: \phi F \epito \rho F$  is a strong epimorphism in $\A$.
\end{prop}
The following proof is set theoretic and makes explicit use of the
fact that $\A$ is algebraic over $\Set$, i.e.~we use that strong
epimorphisms in $\A$ are precisely surjective $T$-algebra morphisms.
In the appendix we provide a purely category theoretic proof, which
is somewhat longer, however. That proof shows that
the above result holds for more general base categories than sets. 
\begin{proof}
  We first prove the following fact:
  \[
    \text{every coalgebra in $\coaf F$ is a regular quotient of some
    coalgebra in $\coafr F$.}
  \]
  Indeed, given any $a: A \to FA$ in $\coaf F$ we know that its
  carrier is a regular quotient of some free $T$-algebra $TX$ with $X$
  finite, via $q: TX \epito A$, say (see
  Remark~\ref{rem:proj}.\ref{rem:proj3}). Since $F$ preserves regular
  epis ($=$ surjections) we can use projectivity of $TX$ (see
  Remark~\ref{rem:proj}.\ref{rem:proj2}) to obtain a coalgebra
  structure $c$ on $TX$ making $q$ an $F$-coalgebra morphism:
  \[
    \xymatrix{
      TX 
      \ar@{-->}[r]^-{c} 
      \ar@{->>}[d]_q
      &
      FTX 
      \ar@{->>}[d]^{Fq}
      \\
      A \ar[r]_-a & FA
    }
  \]
  This implies that we have $\ha c = \ha a \cdot q$. 
  
  Now let $p \in \rho F$. Since $\rho F$ is the colimit of all
  coalgebras in $\coaf F$, we know from Lemma~\ref{lem:strepi} that
  there exists some coalgebra $a: A \to FA$ in $\coaf F$ and $r \in A$
  such that $\ha a (r) = p$. By the above fact, we have $(TX, c)$ in
  $\coafr F$ and the surjective coalgebra morphism $q: TX \epito
  A$. Hence there exists some $s \in TX$ with $q(s) = r$. By the
  finality of $\rho F$ we have the commuting square below:
  \[
    \xymatrix{
      TX \ar@{->>}[r]^-q \ar[d]_{\inj_c} & A \ar[d]^{\ha a} \\
      \phi F \ar@{->>}[r]_-h & \rho F
    }
  \]
  Thus we have $p = \ha a (q(s)) = h(\inj_c(s))$, which shows that $h$
  is surjective as desired. 
\end{proof}

\begin{cor}\label{cor:hepi}
  If $F$ preserves non-empty monomorphisms, then we obtain the situation
  displayed in~\eqref{diag:hierarchy}:
  \[
    \phi F \epito \rho F \epito \theta F \monoto \nu F.
  \]
\end{cor}
Indeed, this follows from Proposition~\ref{prop:hepi} and
Theorem~\ref{thm:rhotheta}. 

\section{A Universal Property of \texorpdfstring{$\phi F$}{phi F}}
\label{sec:univ}

We have seen in Example~\ref{ex:nonproper}(1) that $\phi F$, unlike
$\rho F$ and $\theta F$, does not enjoy a finality property as a
coalgebra. In this section we will prove that, as an algebra for $F$,
$\phi F$ is characterized by a universal property. This property then
determines $\phi F$ uniquely up to isomorphism. To this end we make
the
\begin{asm}
  In addition to Assumptions~\ref{ass:ass} we assume in this section
  that $F$ preserves sifted colimits (cf.~Remark~\ref{rem:sifted}). 
\end{asm}
By Theorem~\ref{thm:phifix}, we know that $\phi F$ is then a fixed
point of $F$ so that by inverting its coalgebra structure we may
regard it as the $F$-algebra $\zeta^{-1}: F(\phi F) \to \phi F$. 

We have already mentioned that both $\rho F$ and $\theta F$ are
characterized by universal properties as $F$-algebras: they are the
initial iterative and initial fg-iterative algebras,
respectively. However, those properties entail that there exists a
unique $F$-coalgebra morphism from every coalgebra in $\coaf F$ to
$\rho F$, and from every coalgebra in $\coafg F$ to $\theta F$,
respectively. That means that simply adjusting the definition of the
notion of iterative algebra does not yield the desired universal
property of $\phi F$, again due to Example~\ref{ex:nonproper}(1).

The key to establishing a universal property of $\phi F$ is to
consider algebras which admit canonical (rather than unique)
coalgebra-to-algebra homomorphisms. The following notion is inspired
by the Bloom algebras introduced by Ad\'amek et al.~\cite{ahm14}.
\begin{defi}
  An \emph{ffg-Bloom algebra} for the functor $F$ is a triple
  $(A, a, \dagger)$ where $a: FA \to A$ is an $F$-algebra and
  $\dagger$ is an operation
  \[
    \frac{TX \xrightarrow{c} FTX, \text{$X$ finite}}{TX \xrightarrow{\sol c}
    A}
  \]
  subject to the following axioms:
  \begin{enumerate}
  \item solution: $\sol c$ is a coalgebra-to-algebra morphism,
    i.e.~the diagram below commutes:
    \[
      \xymatrix{
        TX \ar[r]^-{\sol c}
        \ar[d]_c 
        &
        A
        \\
        FTX 
        \ar[r]_-{F\sol c}
        &
        FA 
        \ar[u]_a
        }
    \]
  \item functoriality: for every coalgebra morphism $m: (TX, c) \to
    (TY,d)$ in $\coafr F$ we have $\sol c = \sol d \cdot m$:
    \[
      \vcenter{
      \xymatrix{
        TX \ar[r]^c \ar[d]_m & FTX \ar[d]^{Fm} \\
        TY \ar[r]_-d & FY
      }}
      \qquad
      \implies
      \qquad
      \vcenter{
      \xymatrix@R-1.5pc{
        TX \ar[rd]^-{\sol c}
        \ar[dd]_m
        \\
        & A \\
        TY \ar[ru]_-{\sol d}
      }}
    \]
  \end{enumerate}
  A \emph{morphism of ffg-Bloom algebras} from $(A,a,\dagger)$ to
  $(B,b,\ddagger)$ is an $F$-algebra morphism preserving solutions,
  i.e.~an $F$-algebra morphism $h: (A,a) \to (B,b)$ such that for
  every $c: TX \to FTX$ in $\coafr F$ we have
  \[
    c^\ddagger = (TX \xrightarrow{\sol c} A \xrightarrow{h} B).
  \]
\end{defi}
\begin{obs}
  The algebra $\zeta^{-1}: F(\phi F) \to \phi F$ together with the
  operation $\ddagger$ given by the colimit injections,
  i.e.~$c^\ddagger = \inj_c: TX \to \phi F$ for every $c: TX \to FTX$ in
  $\coafr F$, clearly is an ffg-Bloom algebra. Indeed, the solution
  axiom holds since $\inj_c$ is a coalgebra morphisms from $(TX, c)$
  to $(\phi F, \zeta)$ and functoriality holds since the $\inj_c$ form a
  compatible cocone of the diagram $D: \coafr F \subto \coa F$.
\end{obs}
\begin{thm}\label{thm:Bloomini}
  The above Bloom algebra on $\phi F$ is the initial ffg-Bloom algebra.
\end{thm}
\begin{proof}
  It remains to prove the universal property. Let $(A, a, \dagger)$
  be any ffg-Bloom algebra. Then the morphisms $\sol c: TX \to A$, for $c: TX
  \to FTX$ ranging over $\coafr F$, form a compatible cocone on the
  diagram $D$ by functoriality. Therefore we have a unique morphism
  $h: \phi F \to A$ such that the triangles below commute
  \[
    \vcenter{
      \xymatrix{
        TX \ar[d]_{\inj_c} \ar[rd]^-{\sol c} \\
        \phi F \ar[r]_-h & A
      }
    }
    \qquad
    \text{for every $c: TX \to FTX$ in $\coafr F$.}
  \]
  In order to see that $h$ is an $F$-algebra morphism consider the diagram
  below:
  \begin{equation}\label{eq:algmor}
    \vcenter{
      \xymatrix{
        \mbox{\hspace{1ex}$TX$}
        \ar[r]^-c
        \ar[d]_{\inj_c}
        & 
        \mbox{$FTX$\ }
        \ar[d]^{F\inj_c}
        \\
        \phi F
        \ar@<2pt>[r]^-{\zeta}
        \ar[d]_h
        &
        F(\phi F)
        \ar@<2pt>[l]^-{\zeta^{-1}}
        \ar[d]^{Fh}
        \\
        A
        \ar@{<-} `l[u] `[uu]^{\sol c} [uu]
        &
        FA 
        \ar[l]^-a
        \ar@{<-} `r[u] `[uu]_{F \sol c} [uu]
      }}
  \end{equation}
  Its outside commutes, for every $c: TX \to FTX$ in $\coafr F$, by
  the solution axiom for $A$, and the left-hand and right-hand
  parts by the definition of $h$. The upper square commutes by the
  solution axiom for $\phi F$. Therefore, for every $c: TX \to FTX$ in
  $\coafr F$ we have
  \[
    h \cdot \inj_c = a \cdot Fh \cdot \zeta \cdot \inj_c. 
  \]
  Use that the colimit injections $\inj_c$ form an epimorphic family
  to conclude that $h$ is an $F$-algebra morphism, i.e.~$h \cdot
  \zeta^{-1} = a \cdot Fh$. This proves existence of a morphism of
  ffg-Bloom algebras from $\phi F$ to $A$. 

  For the uniqueness suppose that $g: \phi F \to A$ is any morphism of
  ffg-Bloom algebras. Then
  \[
    g \cdot \inj_c = g \cdot c^\ddagger = \sol c
  \]
  holds for every $c: TX \to FTX$ in $\coafr F$. Thus, $g=h$ by the
  universal property of the colimit $\phi F$.  
\end{proof}
The following result provides a simple alternative characterization of the
category of ffg-Bloom algebras for $F$ without mentioning 
$\dagger$ and its axioms. This result is similar to~\cite[Prop.~3.4]{ahm14} for
ordinary Bloom algebras. Here $\alg F$ denotes the category of all
$F$-algebras.
\begin{prop}
  The category of ffg-Bloom algebras is isomorphic to the slice
  category $(\phi F, \zeta^{-1})/\!\alg F$.
\end{prop}
\begin{proof}
  (1)~Given an ffg-Bloom algebra $(A, a, \dagger)$, initiality of $\phi F$
  provides an $F$-algebra morphism $h: \phi F \to A$, i.e.~an object of the slice
  category. Moreover, this object assignment clearly gives rise to a
  functor using the initiality of $\phi F$. 

  (2)~In the reverse direction, suppose we are given any $F$-algebra $(A,a)$ and
  $F$-algebra morphism $h: (\phi F, \zeta^{-1}) \to (A,a)$. Then we
  define for every $c: TX \to FTX$ in $\coafr F$, 
  \[
    \sol c = (TX \xrightarrow{\inj_c} \phi F 
    \xrightarrow{h} A).
  \]
  Then using diagram~\eqref{eq:algmor} we see that $\sol c$ satisfies
  the solution axiom: indeed, the outside of the diagram commutes
  since all its inner parts do. Moreover, functoriality of $\dagger$
  follows from that of $\ddagger$: given any $m: (TX,c) \to (TY,d)$ in
  $\coafr F$ we have
  \[
    \sol d \cdot m = h \cdot \inj_d \cdot m = h \cdot \inj_c = \sol c.
  \]
  Furthermore, given a morphism in the slice category, i.e.~we have
  $h: (\phi F, \zeta^{-1}) \to (A,a)$, $g: (\phi F, \zeta^{-1}) \to
  (B,b)$ and $m: (A,a) \to (B,b)$ such that $m \cdot h = g$, we see
  that $m$ is a morphism of ffg-Bloom algebras from $(A,a,\dagger)$ to
  $(B,b,\ddagger)$, where $c^\ddagger: TX
  \to B$ is defined as $g \cdot \inj_c$: indeed, $m$ is an
  $F$-algebra morphism and we have
  \[
    m \cdot \sol c = m \cdot h \cdot \inj_c = g \cdot \inj_c = c^\ddagger. 
  \]
  That this gives a functor from the slice category to the category of
  ffg-Bloom algebras is again straightforward. 

  (3)~We have defined two identity-on-morphisms functors and it
  remains to show that they are mutually inverse on objects.

  From ffg-Bloom algebras to the slice category and back we form for
  the given ffg-Bloom algebra $(A,a,\dagger)$ the ffg-Bloom algebra
  $(A,a,\ddagger)$ where $c^\ddagger = h \cdot \inj_c$ for the unique
  morphism $h: \phi F \to A$ of ffg-Bloom algebras. Hence, since $h$
  preserves solutions we thus have $c^\ddagger = h \cdot \inj_c = \sol
  c$ for every $c: TX \to FTX$ in $\coafr F$. 

  From the slice category to ffg-Bloom algebras and back we take for a
  given $F$-algebra morphism $h: (\phi F, \zeta^{-1}) \to (A,a)$ the
  Bloom algebra $(A,a,\dagger)$ with $\sol c = h \cdot \inj_c$, which
  shows that $h$ is a morphism of ffg-Bloom algebras. Thus, going back
  to the slice category we get back to $h$. 
\end{proof}

\section{Proper Functors and Full Abstractness of \texorpdfstring{$\phi F$}{phi F}}
\label{sec:prop}

In this section we are going to investigate when the three left-hand
fixed points in~\eqref{diag:hierarchy} collapse to one,
i.e.~$\phi F \cong \rho F \cong \theta F$. We introduce proper
functors and show that a functor is proper if and only if $\phi F$ is
fully abstract, i.e.~a subcoalgebra of the final one. This also
entails that the rational fixed point $\rho F$ is fully abstract and at
the same time it is determined by the coalgebras with free finitely
generated carrier. More precisely, the finality of a given locally fp
coalgebra for $F$ can be established by checking the universal
property only for the coalgebras in $\coafr F$
(Corollary~\ref{cor:final}). Here we continue to work under
Assumptions~\ref{ass:ass}.

\begin{rem}
  \label{rem:zigzag}
  \begin{enumerate}
  \item Recall that a \emph{zig-zag} in a category $\A$ is a diagram of the
    form
    \iffull
    \[
      \xymatrix{
        Z_0 \ar[rd]_{f_0} && Z_2 \ar[ld]^{f_1} \ar[rd]_{f_2} &&
        % \hspace*{1em}
        \cdots
        \ar[ld]^{f_3}
        %&
        %\ar@{}[d]|(.6){\objectstyle\cdots}
        %&
        \ar[rd]_{f_{n-2}}
        &&
        Z_{n} 
        \ar[ld]^{f_{n-1}}
        \\
        & Z_1 && Z_3
        &
        \cdots
        % & &
        & Z_{n-1}
      }
    \]
    \else
    \[
      Z_0 \xrightarrow{f_0} Z_1 \xleftarrow{f_1} Z_2 \xrightarrow{f_2} 
      Z_3 \xleftarrow{f_3} \cdots 
      \xrightarrow{f_{n-2}} Z_{n-1} \xleftarrow{f_{n-1}} Z_n.
    \]
    \fi For $\A = \Set^T$, we say that the zig-zag \emph{relates}
    $z_0 \in Z_0$ and $z_n \in Z_n$ if there exist $z_i \in Z_i$,
    $i = 1, \ldots, n-1$ such that $f_i(z_i) = z_{i+1}$ for $i$ even
    and $f_i(z_{i+1}) = z_i$ for $i$ odd.
  \item  \'Esik and Maletti~\cite{em_2010} introduced the notion of a
  proper semiring in order to obtain the decidability of the
  (language) equivalence of weighted automata. A semiring $\S$ is
  called \emph{proper} provided that for every two $\S$-weighted automata
  $A$ and $B$ whose initial states $x$ and $y$, respectively, accept the
  same weighted language there exists a zig-zag 
  \iffull
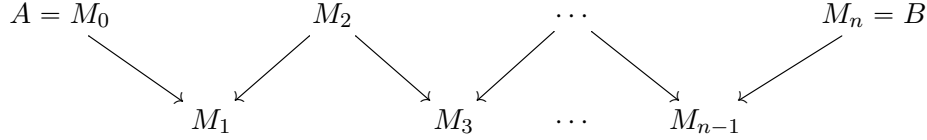
  \[
    \xymatrix{
      A = M_0 \ar[rd] && M_2 \ar[ld] \ar[rd] &&
      % \hspace*{2em} \ar@{}[d]|(.6){\objectstyle\cdots}
      \cdots
      \ar[ld]\ar[rd] && M_{n} = B
      \ar[ld] \\
      & M_1 && M_3 & \cdots & M_{n-1}
    }
  \]
  \else
  \[
    A = M_0 \to M_1 \ot M_2 \to M_3 \ot \cdots \to M_{n-1} \ot M_n = B
  \]
  \fi
  of simulations that \emph{relates} $x$
  and $y$. Recall here that a \emph{simulation} from a weighted
  automaton $(i, (M^a)_a{a \in A}, o)$ with $n$ states to another one
  $(j, (N^a)_{a \in A}, p)$ with $m$ states is an $\S$-semimodule
  morphism represented by an $n \times m$ matrix $H$ over $\S$ such
  that $i \cdot H = j$, $o \cdot H = p$ and $M_a \cdot H = H \cdot
  N_a$. 

  \'Esik and Maletti show that every Noetherian semiring is proper as well
  as the semiring $\Nat$ of natural numbers, which is not
  Noetherian. However, the tropical semiring $(\Nat \cup \{\infty\},
  \min, +, \infty, 0)$ is not proper.
\end{enumerate}
\end{rem}

\medskip Recall from Example~\ref{ex:leading} that $\S$-weighted
automata with input alphabet $\Sigma$ are equivalently coalgebras with
carrier $\S^n$, where $n \geq 1$ is the number of states, for the
functor $FX = \S \times X^\Sigma$ on the category $\smod$. Note that
the $\S^n$ are precisely the free finitely generated $\S$-semimodules,
whence $\S$-weighted automata are precisely the coalgebras in
$\coafr F$, which explains why we are interested in collecting
precisely their behaviour in the form of the fixed point $\phi
F$. Moreover, since simulations of $\S$-weighted automata are clearly
in one to one correspondence with $F$-coalgebra morphisms, one easily
generalizes the notion of a proper semiring as follows. Recall that
$\eta_X: X \to TX$ denotes the unit of the monad $T$.

\begin{defi}\label{dfn:proper}
  We call the functor $F: \A \to \A$ \emph{proper} whenever for every pair of
  coalgebras $c: TX \to FTX$ and $d: TY \to FTY$ in $\coafr F$ and
  every $x \in X$ and $y \in Y$ such that $\eta_X(x) \sim \eta_Y(y)$
  are behaviourally equivalent there exists a zig-zag in $\coafr F$
  relating $\eta_X(x)$ and $\eta_Y(y)$.
\end{defi}

\begin{exa}
  \label{ex:srng}
  A semiring $\S$ is proper iff the functor $FX = \S \times X^\Sigma$
  on $\smod$ is proper for every input alphabet $\Sigma$. We know that
  Noetherian semirings are proper (cf.~Example~\ref{ex:fp=fg}.3), and
  the semiring $\Nat$ of natural numbers is proper. Recently, Sokolova
  and Woracek~\cite{SokolovaW17} have shown that the non-negative
  rationals $\Rat_+$ and non-negative reals $\Real_+$ form proper
  semirings.
\end{exa}

\begin{exa}
  Constant functors are always proper. Indeed, suppose that $F$ is the
  constant functor on some algebra $A$. Then we have $\nu F = A$, and
  for any $F$-coalgebra $B$ its coalgebra structure $c: B \to FB = A$
  is also the unique $F$-coalgebra morphism from $B$ to $\nu F = A$. 

  Now given any $c: TX \to FTX = A$ and $d: TY \to FTY = A$ and
  $x \in TX$, $y \in TY$ as in Definition~\ref{dfn:proper}. Then
  $\eta_X(x) \sim \eta_Y(y)$ is equivalent to
  $c(\eta_X(x)) = d(\eta_Y(y))$. Let $a$ be this element of $A$, and
  extend $x: 1 \to X$, $y: 1 \to Y$ and $a: 1 \to A$ to $T$-algebra
  morphisms $Tx: T1 \to TX$, $Ty: T1 \to TY$ and

  $\kl a: T1 \to A = FT1$ (the latter yielding an $F$-coalgebra). Then 
  
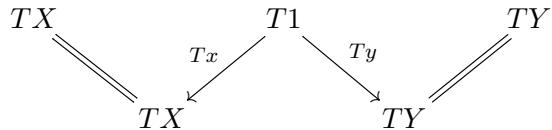
\[
    \xymatrix{
      TX \ar@{=}[rd] && T1 
      \ar[ld]_{Tx} \ar[rd]^{Ty} && TY \ar@{=}[ld] \\
      & TX && TY
      }
  \]
  is the required zig-zag in $\coafr F$ relating $\eta_X(x)$ and $\eta_Y(y)$. 
\end{exa}
\begin{exa}\label{ex:PCAprop}
  Sokolova and Woracek~\cite{SokolovaW17} have recently proved that
  the functor $FX = [0,1] \times X^\Sigma$ on the category $\PCA$ of positively
  convex algebras (see Example~\ref{ex:fp=fg}.4) is proper. In addition, its subfunctor $\hat F$
  given by
  \begin{align*}
    \hat F X = \{ (o, f) \in [0,1] \times X^\Sigma \mid\  
    & 
      \forall s \in \Sigma: \exists p_s\in [0,1], x_s \in X:
      \\
    & o + \sum\limits_{s\in \Sigma} p_s \leq 1, f(s) = p_sx_s\}
  \end{align*}
  is proper. The latter functor was used as coalgebraic type functor
  for the axiomatization of probabilistic systems
  in~\cite{SilvaS11}. In fact, the completeness proof of the
  expression calculus in loc.~cit.~makes use of our
  Corollary~\ref{cor:final} below.
\end{exa}

In general, it seems to be non-trivial to establish that a given
functor is proper (even for the identity functor this may fail; in the
light of Theorem~\ref{thm:main} below this follows from
Example~\ref{ex:nonproper}(1)). However, we will provide in
Proposition~\ref{prop:cong} sufficient conditions on $\A$ and $F$ the
entail properness using our main result:

\begin{thm}\label{thm:main}
  The functor $F$ is proper iff the coalgebra $\phi F$ is a subcoalgebra of $\nu F$. 
\end{thm}
The latter condition states that the unique coalgebra morphism $m: \phi F
\to \nu F$ is a monomorphism in $\A$. 

We present the proof of this theorem in Section~\ref{sec:proof}. Here we
continue with a discussion of the consequences of this result.

\begin{cor}
  \label{cor:Crat}
  If $F$ is proper, then $\phi F$ is the rational fixed point of $F$. 
\end{cor}
\begin{proof}
  Let $u: \rho F \to \nu F$ be the unique $F$-coalgebra morphism. Then
  we have a commutative triangle of $F$-coalgebra morphisms due to
  finality of $\nu F$:
  \iffull
  \[
    \xymatrix{
      \phi F \ar@{->>}[r]^-h 
      & 
      \rho F \ar[r]^-{u} 
      & \nu F.
      \ar@{<->} `u[l] `[ll]_-m [ll]
      }
  \]
  \else
  $
    m = (\phi F \overset{h}{\epito} \rho F \overset{u}{\to} \nu F).
  $
  \fi
  Since $F$ is proper, $m$ is a monomorphism in $\A$, hence so is
  $h$. Since $h$ is also a strong epimorphism by
  Proposition~\ref{prop:hepi}, it is an isomorphism. Thus, $\phi F \cong
  \rho F$ is the rational fixed point of $F$.
\end{proof}

\begin{cor}
  \label{cor:fp=fg}
  Suppose that $F$ preserves non-empty monomorphisms. Then the functor
  $F$ is proper iff
  $\phi F \cong \rho F \cong \theta F \monoto \nu F$.
\end{cor}
\begin{proof}
  If the three fixed points are isomorphic, then $F$ is proper by
  Theorem~\ref{thm:main}.

  Conversely, since $F$ preserves non-empty monomorphisms, we have the
  situation displayed in~\eqref{diag:hierarchy} (see
  Corollary~\ref{cor:hepi}). Now if $F$ is proper we know from
  Corollary~\ref{cor:Crat} that $\phi F \cong \rho F$. Thus, $\rho F$
  is a subcoalgebra of $\nu F$, i.e.~the composition of the last two
  morphisms in~\eqref{diag:hierarchy} is a monomorphism. Thus, so is
  $\rho F \epito \theta F$. Since this is also a strong epimorphism,
  we conclude that $\rho F \cong \theta F$. 
\end{proof}
%Indeed, this follows from Corollary~\ref{cor:Crat} and
%Theorem~\ref{thm:rhotheta}. 
%
Note that this result also entails full abstractness of
$\phi F \cong \rho F$.

A key result for establishing soundness and completeness of coalgebraic regular expression calculi is the following corollary (cf.~\cite[Corollary~3.36]{bms13} and its applications in Sections~4 and 5 of loc.~cit.).
\begin{cor}
  \label{cor:final}
  Suppose that $F$ is proper. Then an $F$-coalgebra $(R, r)$ is a
  final locally fp coalgebra if and only if $(R,r)$ is locally fp and for every
  coalgebra $(TX, c)$ in $\coafr F$ there exists a unique
  $F$-coalgebra morphism from $TX$ to $R$.
\end{cor}
\begin{proof}
  The implication ``$\Rightarrow$'' clearly holds

  For ``$\Leftarrow$'' it suffices to prove that for every
  $a: A \to FA$ in $\coaf F$ there exists a unique $F$-coalgebra
  morphism from $A$ to $R$. In fact, it then follows that $R$ is the
  final locally fp coalgebra. To see this write an arbitrary locally fp coalgebra
  $A$ as a filtered colimit of a diagram $D: \D \to \coaf F \subto
  \coa F$ with colimit
  injections $h_d: Dd \to A$ ($d$ an object in $\D$). Then the
  unique $F$-coalgebra morphisms $u_d: Dd \to R$ form a compatible
  cocone, and so one obtains a unique $u: A \to R$ such that
  $u \cdot h_d = u_d$ holds for every object $d$ of $\D$. It is now
  straightforward to prove that $u$ is a unique $F$-coalgebra morphism
  from $A$ to $R$.

  Now let $a: A \to FA$ be a coalgebra in $\coaf F$. For every
  $(TX,c)$ in $\coafr F$ denote by $c^\ddagger: TX \to R$ the unique
  $F$-coalgebra morphism that exists by assumption. These morphisms $c^\ddagger$
  form a compatible cocone of the diagram $\coafr F \subto \coa
  F$. Thus, we obtain a unique $F$-coalgebra morphism
  $m': \phi F \cong \rho F \to R$ such that the following diagram
  commutes for every $c: TX \to FTX$ in $\coafr F$:
  \[
    \xymatrix{
      TX
      \ar[d]_{\inj_c} \ar[rd]_(.4){\ha c} \ar[rrd]^{c^\ddagger} \\
      \phi F \ar@{=}[r]_-{\cong} & \rho F \ar[r]_{m'} & R
    }
  \]
  Therefore we have an $F$-coalgebra morphism 
  \[
    h = (A \xrightarrow{\ha a} \rho F \xrightarrow{m'} R). 
  \]
  To prove it is unique, assume that $g: A \to R$ is any $F$-coalgebra
  morphism. As in the proof of Proposition~\ref{prop:hepi}, we know
  that $A$ is the quotient of some $TX$ in $\coafr F$ via $q:
  TX \epito A$, say. Then we have 
  \iffull\[\else\/$\fi
    m' \cdot \ha a \cdot q = g \cdot q
  \iffull\]\else\/$ \fi
  because there is only one $F$-coalgebra morphism from $TX$ to $R$ by
  hypothesis. It follows that $h = m' \cdot \ha a = g$ since $q$ is
  epimorphic. 
\end{proof}

The next result provides sufficient conditions for properness of
$F$. It can be seen as a category-theoretic generalization of \'Esik's
and Maletti's result~\cite[Theorem~4.2]{em_2010} that Noetherian
semirings are proper.
\begin{prop}
  \label{prop:cong}
  Suppose that finitely generated algebras in $\A$ are closed under
  kernel pairs and that $F$ maps kernel pairs to weak pullbacks in
  $\Set$. Then $F$ is proper.
\end{prop}
\takeout{% Proof based on Appendix replaced by new one
\begin{proof}
  From Proposition~\ref{prop:coeq} and Corollary~\ref{cor:colim} we
  know that $\phi F \cong \rho F$. Furthermore, since $F$ maps kernel pairs
  to weak pullbacks in $\Set$ we see that $F$ preserves monomorphisms;
  indeed, $m: A \monoto B$ is a mono in $\A$ iff and only if its
  kernel pair is $\id_A,\id_A$. Thus $F\id_A, F\id_A$ form a weak
  pullback in $\Set$, which is in fact a pullback, whence $Fm$ is
  monomorphic.
  
  By Lemma~\ref{lem:fg=fp}, it follows that finitely generated objects
  are finitely presentable. Therefore, by Proposition~\ref{prop:sub},
  $\rho F$ and thus $\phi F$ is a subcoalgebra of $\nu F$, whence $F$ is
  proper by Theorem~\ref{thm:main}.
\end{proof}
}% end takeout
\begin{proof}
  First, since $F$ maps kernel pairs
  to weak pullbacks in $\Set$ we see that $F$ preserves monomorphisms;
  indeed, $m: A \monoto B$ is a mono in $\A$ iff and only if its
  kernel pair is $\id_A,\id_A$. Thus $F\id_A, F\id_A$ form a weak
  pullback in $\Set$, which is in fact a pullback, whence $Fm$ is
  monomorphic.

  Now let $(TX, c)$ and $(TY, d)$ be in $\coafr F$, $x \in X$ and $y \in
  Y$ such that $\fin c(\eta_X(x)) = \fin d(\eta_Y(y))$. It is our task
  to construct a zig-zag relating $\eta_X(x)$ and $\eta_Y(y)$.
  
  Form $Z = X + Y$ and let $e: TZ \to FTZ$ be the coproduct of the
  coalgebras $(TX, c)$ and $(TY,d)$ in $\coafr F$ (see
  Lemma~\ref{lem:coprod}). Take the factorization of
  $\fin e: TZ \to \nu F$ into a strong epi $q: TZ \epito A$ followed
  by a monomorphism $m: A \monoto \nu F$. Since $F$ preserves
  non-empty monos, we obtain a unique coalgebra structure
  $a: A \to FA$ such that $q$ and $m$ are coalgebra morphisms (see
  Remark~\ref{rem:fs}(\ref{i:fs2})). Now take the kernel pair
  $f,g: K \parallel TZ$ of $q$. Since $TZ$ and its quotient $A$ are
  finitely generated $T$-algebras, so is $K$ because finitely
  generated $T$-algebras are closed under taking kernel pairs by
  assumption. Now $F$ maps the kernel pair $f,g$ to a weak pullback
  $Ff, Fg$ of $Fq$ along itself in $\Set$. Thus, we have a map
  $k\colon K \to FK$ such that the diagram below commutes:
  \begin{equation}
    \label{diag:kappa}
    \vcenter{
    \xymatrix{
      K
      \ar@<-.5ex>[d]_f
      \ar@<.5ex>[d]^g
      \ar@{-->}[r]^-k
      &
      FK
      \ar@<.5ex>[d]^{Fg}
      \ar@<-.5ex>[d]_{Ff}
      \\
      TZ
      \ar[d]_q
      \ar[r]^-e
      &
      FTZ
      \ar[d]^{Fq}
      \\
      A
      \ar[r]_-a
      &
      FA
      }}
  \end{equation}
  Notice that we do not claim that $k$ is a $T$-algebra
  morphism. However, since $K$ is a finitely generated $T$-algebra,
  it is the quotient of some free finitely generated $T$-algebra $TR$
  via $p\colon TR \epito K$, say. Now we choose some splitting
  $s\colon K \to TR$ of $p$ in $\Set$, i.\,e., $s$ is a map such that
  $p \o s = \id$. Next we extend the map
  $r_0 = Fs \o k \o p \o \eta_R$ to a $T$-algebra morphism
  $r\colon TR \to FTR$; it follows that the outside of the diagram
  below commutes:
  \begin{equation}
    \label{diag:d}
    \vcenter{
      \xymatrix{
        R
        \ar[d]_{\eta_R}
        \ar[rd]^{r_0}
        \\
        TR
        \ar@{-->}[r]^-r
        \ar[d]_p
        &
        FTR
        \ar[d]^{Fp}
        \\
        K
        \ar[r]_-k
        &
        FK
      }
    }
  \end{equation}
  (Notice that to obtain $r$ we cannot simply use projectivity of $TR$
  since $k$ is not necessarily a $T$-algebra homomorphism.)

  We do not claim that this makes $p$ a coalgebra morphism (i.\,e., we
  do not claim the lower square in~\eqref{diag:d} commutes). However,
  $f\o p$ and $g \o p$ are coalgebra morphisms from $(TR, r)$ to
  $(TZ, e)$; in fact, to see that
  $$
  e \o (f \o p) =  F(f \o p) \o r
  $$
  it suffices that this equation of $T$-algebra morphisms holds
  when both sides are precomposed with $\eta_R$. To this end we compute
  $$
  \begin{array}{rcl@{\qquad}p{3cm}}
    e \o f \o p \o \eta_R & = & Ff \o k \o p \o \eta_R &
    see~\eqref{diag:kappa}, \\
    & = & Ff \o Fp \o r_0 & outside of~\eqref{diag:d},
    \\
    & = & Ff \o Fp \o r \o \eta_R & definition of $d$.
  \end{array}
  $$
  Similarly, $g \o p$ is a coalgebra morphism.

  Now consider the following zig-zag in $\coafr F$ (recall that the
  algebra $TZ$ is the coproduct of $TX$ and $TY$ with coproduct
  injections $T\inl$ and $T\inr$):
  \[
    \xymatrix{
      TX \ar[rd]_-{T\inl} 
      && 
      TR
      \ar[ld]_-{f \cdot p}
      \ar[rd]^-{g \cdot p}
      && 
      TY
      \ar[ld]_-{T\inr}
      \\
      & TZ && TZ 
      }
  \]
  We now show that this zig-zag relates $\eta_X(x)$ and
  $\eta_Y(y)$. Let $x' = T\inl(\eta_X(x))$ and $y' =
  T\inr(\eta_Y(y))$. Then we have
  \[
    \fin e (x') = \fin e \cdot T\inl(\eta_X(x)) = \fin c (\eta_X(x)) =
    \fin d (\eta_Y(y)) = \fin e \cdot T\inr(\eta_Y(y)) = \fin e (y').
  \]
  Hence, since $\fin e = m \cdot q$ and $m$ is monomorphic, we obtain
  $q(x') = q(y')$. Thus, there exists some $k \in K$ such that
  $f(k) = x'$ and $g(k) = y'$ by the universal property of the kernel
  pair. Finally, since $p: TR \epito K$ is surjective we obtain some
  $z \in TR$ such that $p(z) = k$ whence $f\cdot p(z) = x'$ and
  $g \cdot p(z) = y'$. This completes the proof.
\end{proof}
\begin{rem}
\begin{enumerate}
\item Note that closure of finitely generated algebras under kernel pairs
can equivalently be stated in general algebra terms as follows: every
congruence $R$ of a finitely generated algebra $A$ is finitely
generated as a subalgebra $R \subto A \times A$ (observe that this is
\emph{not} equivalent to stating that $R$ is a finitely generated
congruence).

\item For a lifting $F$ of a set functor $F_0$, the condition that $F$
  maps kernel pairs to weak pullbacks in $\Set$ holds whenever $F_0$
  preserves weak pullbacks. Hence, all the functors on algebraic
  categories mentioned in Example~\ref{ex:rat} satisfy this
  assumption.

\item For the special case of a lifting, a variant of the argument in
  the proof of Proposition~\ref{prop:cong} was used
  in~\cite[Proposition~3.34]{bms13} in order to prove that every
  coalgebra in $\coaf F$ is a coequalizer of a parallel pair of
  morphisms in $\coafr F$. This has inspired
  Winter~\cite[Proposition~7]{Winter15} who uses a very similar
  argument to prove that, for a distributive law $\lambda$,
  $\lambda$-bisimulations (see Bartels~\cite{bartels_thesis}) are
  sound and complete for $\lambda$-bialgebras (see
  Remark~\ref{rem:bialg}). It turns out that, for a lifting $F$,
  Proposition~\ref{prop:cong} is a consequence of Winter's result, or,
  in other words, our result can be understood as a slight
  generalization of Winter's one.
\end{enumerate}
\end{rem}
\begin{exas}
  \begin{enumerate}
  \item The first condition in Proposition~\ref{prop:cong} is not
    necessary for properness of $F$. In fact, it fails in the category
    of semimodules for $\Nat$, viz.~the category of commutative
    monoids: in fact, consider the finitely generated commutative
    monoid $\Nat \times \Nat$ and its submonoid infinitely generated
    by 
    \[
      \{ (n, n+1) \mid n \in \Nat\}, 
    \]
    which is easily seen not be finitely generated. However, as we mentioned in
    Example~\ref{ex:srng}, $FX = \Nat \times X^\Sigma$ is proper on
    the category of commutative monoids. 
  \item In Example~\ref{ex:fp=fg}(\ref{ex:fp=fg4}) we mentioned that, in the category
    $\PCA$ of positively convex algebras, fg- and fp-objects
    coincide. However, fg-objects are not closed under kernel
    pairs. In fact, the interval $[0,1]$ is the free positively convex
    algebra on two generators, but
    $\{(0,0), (1,1)\} \cup (0,1) \times (0,1)$ is a congruence on
    $[0,1]$ that is not an fg-object (i.e.~a
    polytope)~\cite[Example~4.13]{SokolovaW15}.  Thus, properness of
    the functors in Example~\ref{ex:PCAprop} does not follow from
    Proposition~\ref{prop:cong}.
  \end{enumerate}
\end{exas}

\section{Proof of Theorem~\ref{thm:main}}
\label{sec:proof}

In this section we will present the proof of our main result
Theorem~\ref{thm:main}. We start with two technical lemmas.

\begin{rem}
  \label{rem:pp}
  Recall~\cite[Proposition~11.28.2]{arv} that every free $T$-algebra
  $TX$ is \emph{perfectly presentable}, i.e.~the hom-functor
  $\Set^T(TX, -)$ preserves sifted colimits
  (cf.~Remark~\ref{rem:sifted}). It follows that for every sifted
  diagram $D: \D \to \Set^T$ and every $T$-algebra morphism
  $h: TX \to \colim D$ there exists some $d \in \D$ and
  $h': TX \to Dd$ such that 
  \[
    \xymatrix{
      & Dd \ar[d]^{\inj_d} \\
      TX \ar@{-->}[ru]^-{h'} \ar[r]_-h & \colim D.
    }
  \]
\end{rem}

\begin{lem}
  \label{lem:fac}
  For every finite set $X$ and map $f: X \to \phi F$ there exists an object
  $(TY,d)$ in $\coafr F$ and a map $g: X \to Y$ such that \iffull the
  triangle below commutes: 
  \[
    \xymatrix{
      &&
      X 
      \ar[d]^{f} 
      \ar[lld]_-{g}
      \\
      Y \ar[r]_-{\eta_Y} & TY \ar[r]_-{\inj_d} & \phi F
      }
  \]
  \else
  $
    f = (X \xrightarrow{g} Y \xrightarrow{\eta_Y} TY \xrightarrow{\inj_d} \phi F).
  $
  \fi
\end{lem}
\iffull
\begin{proof}
  We begin by extending $f$ to a $T$-algebra morphism $h = \kl f: TX \to
  \phi F$. By Remark~\ref{rem:pp}, there exists some $c: TZ \to FTZ$ in
  $\coafr F$ and a $T$-algebra morphism $h': TX \to TZ$ such that $h =
  \inj_c \cdot h'$. Let $f' = h' \cdot \eta_X$, let $Y = X+Z$ and
  consider the $T$-algebra morphism $\kl{[f',\eta_Z]}: TY \to
  TZ$. This is a split epimorphism in $\Set^T$; we have $T\inr: TZ \to
  TY$ with
  \[
    \kl{[f',\eta_Z]} \cdot T\inr = \kl{\eta_Z} = \id_{TZ},
  \]
  where the last equation follows from the uniqueness property of
  $\kl{\eta_Z}$ (see Section~\ref{sec:algcoalg})
  by the laws of $\kl{(-)}$. We therefore get a coalgebra structure
  \[
    d = (TY \xrightarrow{\kl{[f',\eta_Z]}} TZ \xrightarrow{c} 
    FTZ \xrightarrow{T\inr} FTY)
  \]
  such that $\kl{[f',\eta_Z]}$ is an $F$-coalgebra morphism from
  $(TY, d)$ to $(TZ,c)$. Since $Y$ is a finite set,
  $(TY, d)$ is an $F$-coalgebra in $\coafr F$, and hence
  $\inj_c \cdot \kl{[f',\eta_Z]} = \inj_d$. Thus we see that
  $g = \inl: X \to Y$ is the desired morphism due to the commutative
  diagram below:
  \[
    \xymatrix@C+1pc{
      &&& X
      \ar `l[llld]_{g=\inl} [llld]
      \ar[ld]^{f'} 
      \ar[d]^f 
      \\
      Y
      \ar[r]_-{\eta_Y}
      \ar@/^1pc/[rr]^-{[f',\eta_Z]} 
      &
      TY
      \ar[r]_-{\kl{[f',\eta_Z]}}
      &
      TZ
      \ar[r]_-{\inj_c}
      &
      \phi F
      \ar@{<-} `d[l] `[ll]^-{\inj_d} [ll]
    }
  \]
\end{proof}
\fi
\begin{rem}
  \label{rem:setcolimit}
  Recall that a colimit of a diagram $D: \D \to \Set$ is computed as
  follows:
  \[
    \colim D = \big(\!\coprod\limits_{d \in \D} Dd\big)/\mathord{\sim},
  \]
  where $\sim$ is the least equivalence on the coproduct (i.e.~the
  disjoint union) of all $Dd$ with $x \sim Df(x)$ for every
  $f:d \to d'$ in $\D$ and every $x \in Dd$. In other words, for every
  pair of objects $c, d$ of $\D$ and $x \in Dc$, $y \in Dd$ we have
  $x \sim y$ iff there is a zig-zag in $\D$ whose $D$-image \iffull
  \[
    \xymatrix{
      Dc = Dz_0 \ar[rd]_{Df_0} && Dz_2 \ar[ld]^{Df_1} \ar[rd]_{Df_2} && \cdots
      \ar[ld]^{Df_3} \ar[rd]_{Df_{n-2}} && z_{n} = Dd 
      \ar[ld]^{Df_{n-1}} \\
      & Dz_1 && Dz_3 & \cdots & Dz_{n-1}
    }
  \]
  \fi
  relates $x$ and $y$ (cf.~Remark~\ref{rem:zigzag}). 
\end{rem}
\begin{lem}
  \label{lem:zigzag}
  Let $(TX, c)$ and $(TY, d)$ be coalgebras in $\coafr F$, 
  $x \in TX$, and $y \in TY$. Then the following are equivalent:
  \begin{enumerate}
  \item $\inj_c(x) = \inj_d(y) \in \phi F$, and
  \item there is a zig-zag in $\coafr F$ relating $x$ and $y$. 
  \end{enumerate}
\end{lem}
\begin{proof}
  By Remark~\ref{rem:sifted}(\ref{rem:sifted3}), $\phi F$ is a sifted colimit. Hence, the
  forgetful functor $\coa F \to \Set^T \to \Set$ preserves this
  colimit. Thus the colimit $\phi F$ is formed as recalled in
  Remark~\ref{rem:setcolimit}:
  \[
    \phi F \cong \big(\coprod_c TX_c\big)/\mathord{\sim},
  \]
  where $c: TX_c \to FTX_c$ ranges over the objects of $\coafr
  F$. Therefore, we have the desired equivalence.
\end{proof}
\begin{proof}[Proof of Theorem~\ref{thm:main}]
  ``$\Rightarrow$'' Suppose that for $m: \phi F \to \nu F$ we have $x, y
  \in \phi F$ with $m(x) = m(y)$. We apply Lemma~\ref{lem:fac} to 
  \iffull
  \[
    1 \xrightarrow{x} \phi F \qquad\text{and}\qquad 1 \xrightarrow{y} \phi F,
  \]
  \else
  $1 \xrightarrow{x} \phi F$ and $1 \xrightarrow{y} \phi F$, \fi
  respectively, to obtain two objects
  $c:TX \to FTX$ and $d: TY \to FTY$ in $\coafr F$ with $x' \in X$ and
  $y' \in Y$ such that $\inj_c(\eta_X(x')) = x$ and
  $\inj_d(\eta_Y(y')) = y$. By the uniqueness of coalgebra morphisms
  into $\nu F$ we have 
  \begin{equation}
    \label{eq:coalgmor}
    \fin{c} = m \cdot \inj_c 
    \qquad\text{and}\qquad 
    \fin{d} = m \cdot \inj_d. 
  \end{equation}
  Thus we compute:
  \[
    \fin c(\eta_X(x')) = m \cdot \inj_c\cdot \eta_X(x') = m(x) 
    = 
    m(y) = m \cdot \inj_d \cdot \eta_Y(y') = \fin d(\eta_Y(y')). 
  \]
  Since $F$ is proper by assumption, we obtain a zig-zag in $\coafr F$
  relating $\eta_X(x')$ and $\eta_Y(y')$. By Lemma~\ref{lem:zigzag},
  these two elements are merged by the colimit injections, and we have
  $x = \inj_c(\eta_X(x')) = \inj_d(\eta_Y(y') = y$. We conclude that
  $m$ is monomorphic.

  ``$\Leftarrow$'' Suppose that $m: \phi F \monoto \nu F$ is a
  monomorphism. Let $c: TX \to FTX$ and $d: TY \to FTY$ be objects of
  $\coafr F$, and let $x \in X$ and $y \in Y$ be such that $\fin c
  (\eta_X(x)) = \fin d(\eta_Y(y))$. Using~\eqref{eq:coalgmor} and the
  fact that $m$ is monomorphic we get $\inj_c(\eta_X(x)) =
  \inj_d(\eta_Y(y))$. By Lemma~\ref{lem:zigzag}, we thus obtain a
  zig-zag in $\coafr F$ relating $\eta_X(x)$ and $\eta_Y(y)$. This
  proves that $F$ is proper.
\end{proof}

\section{Conclusions and Further Work}
\label{sec:con}
%\iffull\else\enlargethispage{10pt}\fi
\takeout{% Stoffsammlung
\begin{itemize}
\item TODO: list open problems from p.~6 of the notes!
\item the discussion on completeness proofs on p.~9f in the notes,
  i.e.~that completeness proof in~\cite{bms13} extends to $\S = \Nat$
\item proper functors on convex sets
\end{itemize}
}% end takeout

Inspired by \'Esik and Maletti's notion of a proper semiring, we have
introduced the notion of a proper functor. We have shown that, for a
proper endofunctor $F$ on an algebraic category preserving regular
epis and monos, the rational fixed point $\rho F$ is fully abstract
and moreover determined by those coalgebras with a free finitely
generated carrier (i.e.~the target coalgebras of generalized
determinization).

Our main result also shows that properness is necessary for this kind
of full abstractness. For categories in which fg-objects are closed
under kernel pairs we saw that when $F$ maps kernel pairs to weak
pullbacks in $\Set$, then it is proper. This provides a number of
examples of proper functors. However, in several categories of
interest the condition on kernel pairs fails, e.g.~in
$\Nat$-semimodules (commutative monoids) and positively convex
algebras. There can still be proper functors, e.g.
$FX = \Nat \times X^\Sigma$ on the former and
$FX = [0,1] \times X^\Sigma$ on the latter. But establishing
properness of a functor without using Proposition~\ref{prop:cong}
seems non-trivial, and we leave the task of finding more examples of
proper functors for further work. 
 
One immediate consequence of our results is that the soundness and
completeness proof for the expression calculi for weighted
automata~\cite{bms13} extends from Noetherian to proper semirings.
In fact, \'Esik and Kuich~\cite[Theorems~7.1 and 8.5]{ek_2012} already
provide sound and complete axiomatizations of weighted language
equivalence for (certain subclasses of) proper semirings $\S$ by showing
that $\S$-rational weighted languages form certain free algebras. 

In the future, when additional proper functors are known, it will be
interesting to study regular expression calculi for their coalgebras
and use the technical machinery developed in the present paper for
soundness and completeness proofs.

Another task for future work is to study the new fixed point $\phi F$
in its own right. Here we have already proven that $\phi F$ is
characterized uniquely (up to isomorphism) as the initial ffg-Bloom
algebra. In the future, it might be interesting to investigate
\emph{free} (rather than initial) ffg-Bloom algebras. Moreover,
related to ordinary Bloom algebras~\cite{ahm14} there is the notion of
an Elgot algebra~\cite{amv_elgot}. It is known that for every object
$Y$ of an lfp category, the parametric rational fixed point
$\rho (F(-) + Y)$ yields a free Elgot algebra on $Y$. In addition, the
category of algebras for the ensuing monad is isomorphic to the
category of Elgot algebras for $F$. In~\cite{amu18_cmcs}, the new
notion of an ffg-Elgot algebra for $F$ is introduced, and it is shown
that for free finitely generated algebras $Y$ the parametric fixed
point $\phi(F(-) + Y)$ forms a free ffg-Elgot algebra for $F$ on $Y$,
and furthermore the category of ffg-Elgot algebras for $F$ is monadic
over our algebraic base category $\A$. It remains an open question
whether ffg-Elgot algebras (or ffg-Bloom algebras) are monadic over
$\Set$.

%\vspace*{-10pt}
% 
% Bibliography
%
%\bibliographystyle{plainurl}% the recommended bibstyle
\iffull
\bibliographystyle{abbrv}
\bibliography{ourpapers,coalgebra}
\else
\bibliographystyle{plainurl}
\bibliography{refs}
\fi

\clearpage
\appendix

\section*{Appendix: Category Theoretic Proof of Proposition~\ref{prop:hepi}}

Note first that for every $c: TX \to FTX$ in $\coaf F$ we clearly have
\[
\ha c = (TX \xrightarrow{\inj_c} \phi F \xrightarrow{h} \rho F)
\]
by the finality of $\rho F$. Recall that for strong epis the same
cancellation law as for epis holds: if $e \cdot e'$ is a strong epi, then
so is $e$; a similar law holds for strongly epimorphic families.
Hence, we are done if we show that the $\ha c$ where
$c: TX \to FTX$ ranges over $\coafr F$ forms a jointly strongly
epimorphic family, too. This is done by using that the $\ha a$, where
$a: A \to FA$ ranges over $\coaf F$, form a strongly epimorphic family
(to see this use Lemma~\ref{lem:strepi} once again). 

The key observation is as follows: given any $a: A \to FA$ in
$\coaf F$ we know that its carrier is a regular quotient of some free
$T$-algebra $TX$ with $X$ finite, via $q: TX \epito A$, say. Since $F$
preserves regular epis ($=$ surjections) we can use projectivity of
$TX$ (see Remark~\ref{rem:proj}(\ref{rem:proj2})) to obtain a coalgebra
structure $c$ on $TX$ making $q$ an $F$-coalgebra morphism:
%\iffull
\[
  \xymatrix{
    TX 
    \ar@{-->}[r]^-{c} 
    \ar@{->>}[d]_q
    &
    FTX 
    \ar@{->>}[d]^{Fq}
    \\
    A \ar[r]_-a & FA
  }
\]
%\else
%indeed, we have the regular epi $Fq: FTX \epito FA$ and the morphism
%$a \cdot q: TX \to FA$; thus we obtain $c: TX \to FTX$ with $Fq \cdot
%c = a \cdot q$ as desired. \fi 
This implies that we have $\ha c = \ha a \cdot q$. 

Now suppose that we have two parallel morphisms $f, g$ such that for every $c:
TX \to FTX$ in $\coafr F$ we have $f \cdot \ha c = g \cdot \ha
c$. Then for every $a: A \to FA$ in $\coaf F$ we obtain
\[
f \cdot \ha a \cdot q = f \cdot \ha c = g \cdot \ha c = g \cdot \ha a
\cdot q,
\]
which implies that $f \cdot \ha a = g \cdot \ha a$ since $q$ is
epimorphic. Hence $f = g$ since the $\ha a$ form a jointly epimorphic
family. This proves that the $\ha c$ form a jointly epimorphic
family. 

To see that they form a strongly jointly epimorphic family, assume
that we are given a monomorphism $m: M \monoto N$ and morphisms $g:
\rho F \to N$ and $f_c: TX \to M$ for every $c: TX \to FTX$ in $\coafr
F$ such that $m \cdot f_c  = g \cdot \ha c$. We extend the family
$(f_c)$ to one indexed by all $a: A \to FA$ in $\coaf F$ as
follows. We have that any such $(A,a)$ is a quotient coalgebra of some
$(TX, c)$ via $q: TX \epito A$, which is the coequalizer of some
parallel pair $k_1, k_2: K \to TX$ in $\A$. Thus we have
\begin{align*}
  m \cdot f_c \cdot k_1 
  &= g \cdot \ha c \cdot k_1 \\
  &= g \cdot \ha a \cdot q \cdot k_1 \\
  &= g \cdot \ha a \cdot q \cdot k_2 \\
  &= g \cdot \ha c \cdot k_2 \\
  &= m \cdot f_c \cdot k_2,
\end{align*} 
which implies that $f_c \cdot k_1 = f_c \cdot k_2$ since $m$ is
monomorphic. Therefore we obtain a unique $f_a: A \to M$ such that
$f_a \cdot q = f_c$ using the universal property of the coequalizer
$q$. Hence we can compute
%\iffull\[\else\/$\fi
\[
  m \cdot f_a \cdot q = m \cdot f_c = g \cdot \ha c = g \cdot \ha a
  \cdot q,
\]
%\iffull\]\else\/$ \fi
which implies $m \cdot f_a = g \cdot \ha a$ since $q$ is
epimorphic. Now we use that the $\ha a$ are jointly strongly
epimorphic (cf.~Lemma~\ref{lem:strepi}) to obtain a unique morphism
$d: \rho F \to M$ with $d \cdot \ha a = f_a$ and $m \cdot d = g$ for
all $a: A \to FA$ in $\coaf F$. In particular, $d$ is the desired
diagonal fill-in since $\coafr F$ is a full subcategory of $\coaf F$. As for
the uniqueness of the fill-in $d$ we still need to check that any $d$
with $d \cdot \ha c = f_c$ for all $c: TX \to FTX$ in $\coafr F$ and
$m \cdot d = g$ also fulfils $d \cdot \ha a = f_a$ for every
$a: A \to FA$ in $\coaf F$. Indeed, this follows from
\[
d \cdot \ha a \cdot q = d \cdot \ha c = f_c = f_a \cdot q
\]
using that $q$ is epimorphic. \qed

\end{document}